\documentclass{article}
\usepackage[a4paper,hmargin=1.3in,vmargin=1.5in]{geometry} 
\usepackage{amsthm}
\usepackage{caption}
\usepackage{subcaption}
\usepackage{amsmath, amssymb}
\usepackage{xcolor}
\usepackage{tikz,longtable,float}
\usetikzlibrary{calc}
\usetikzlibrary{patterns}
\newtheorem*{theorem*}{Theorem}
\newtheorem{theorem}{Theorem}

\newtheorem{lemma}{Lemma}
\newtheorem{definition}{Definition}

\usepackage{url, hyperref}

\def \endproof {\hfill $\Box$ \smallskip}

\DeclareMathOperator{\Conv}{Conv}
\DeclareMathOperator{\Stack}{Stack}
\DeclareMathOperator{\sign}{sign}
\DeclareMathOperator{\maxdeg}{\Delta_G}

\DeclareMathOperator{\pr}{pr}
\DeclareMathOperator{\Area}{Area}
\DeclareMathOperator{\Vol}{Vol}

\newcommand{\gridsize}{L}

\renewcommand{\P}{\mathcal{P}}
\newcommand{\ns}{\!\!\!\!}
\newcommand{\R}{\mathbb{R}}
\newcommand{\Z}{\mathbb{Z}}

\begin{document}




\title{A Duality Transform for Constructing\\ Small Grid Embeddings of 3d Polytopes}

\author{Alexander Igamberdiev\thanks{LG Theoretische Informatik, FernUniversit\"at in Hagen, Hagen, Germany.
{\tt [alexander.igamberdiev|andre.schulz]@fernuni-hagen.de}.
This work was funded by the German Research Foundation (DFG) under grant SCHU 2458/2-1.} 
\and Andr\'e Schulz\footnotemark[1]}

\maketitle

\begin{abstract}
We study the problem of how to obtain an integer realization of a 3d polytope when
an integer realization of its dual polytope is given. We focus on grid embeddings with
small coordinates and develop novel techniques
based on Colin de Verdi\`ere matrices and the Maxwell--Cremona lifting 
method.

We show that every truncated 3d polytope
with $n$ vertices 
can be realized  on a grid of size $O(n^{9\log6+1})$.
Moreover, for every simplicial 3d polytope with $n$ vertices with maximal vertex degree $\Delta$
and vertices placed on  an $L\times L \times L$ grid,
a dual polytope can be realized on an integer grid of size $O(nL^{3\Delta+9})$.
This implies that for a class $\mathcal{C}$  of  
simplicial 3d polytopes with 
bounded vertex degree 
and  polynomial size grid embedding, the
dual polytopes of $\mathcal{C}$ can
be realized on a polynomial size grid as well. 
\end{abstract}

\maketitle

\section{Introduction}

By Steinitz's theorem the graphs of convex 3d polytopes\footnote{In our terminology polytopes are always considered \emph{convex}.} are exactly 
the planar 3-connected graphs~\cite{S16}. Several methods are known for 
realizing a planar 3-connected graph $G$ as a polytope with graph $G$  on the grid~\cite{DG97,EG95,RRS11,OS94,RG96,Sch11}.
It is challenging to find algorithms that produce  polytopes with
small integer coordinates. Having a realization with small grid size is a desirable
feature, since then the polytope can be stored and processed efficiently. Moreover, grid embeddings imply 
good vertex and edge resolution. Hence, they produce ``readable'' drawings.

In 2d, every planar 3-connected graph with $n$ vertices
can be drawn with straight-line edges on an $O(n)\times O(n)$ grid 
without crossings~\cite{FPP90}, and a drawing with convex faces
can be realized on an $O(n^{3/2}\times n^{3/2})$ grid~\cite{BR06}. For the 
realization as a polytope the currently best algorithm guarantees an integer
embedding with coordinates of size at most $O(147.7^n)$~\cite{BS10,RRS11}.
The current best lower bound is $\Omega(n^{3/2})$~\cite{A61}. Closing
this  gap is an intriguing  open problem
in lower dimensional polytope theory. 

Recently, progress has been made for
a special class of 3d polytopes, the so-called \emph{stacked polytopes}.
A \emph{stacking operation} replaces a triangular face of a polytope
with a tetrahedron, while maintaining the convexity of the embedding (see Fig.~\ref{fig:stackedTruncated}).
A polytope that can be constructed from a tetrahedron and a sequence of stacking
operation is called a stacked 3d polytope, or for the scope of this paper simply a \emph{stacked polytope}. 
The graphs of stacked polytopes are planar 3-trees.
Stacked polytopes can be embedded 
on a grid that is polynomial in $n$~\cite{DS11}. This is, however, the only nontrivial polytope 
class for which such an algorithm is known.

\begin{figure}[t]
\centering
\begin{subfigure}[t]{0.4\linewidth}
\centering
\begin{tikzpicture}[line width=1pt]
	\begin{scope}[scale=1.2, x={(1cm,0cm)}, y={(0.25cm,0.17cm)}, z={(0cm,0.5cm)}]
			\path (0,0,-1) node[name=justtoshift, shape=coordinate]{};
	
			\path (-2,1,0) node[name=s1, shape=coordinate]{};
			\path (2,1,0) node[name=s2, shape=coordinate]{};
			\path (0,-2,0) node[name=s3, shape=coordinate]{};
			\path ($2/5*(s1)+2/5*(s2)+1/5*(s3)+(0,0,3)$) node[name=s4, shape=coordinate]{};
			
			\path ($1/2*(s2)+1/4*(s3)+1/4*(s4)+(0,0,0.4)$) node[name=s234, shape=coordinate]{};
		
			\path[draw] (s4)--(s1);
			\path[draw] (s4)--(s2);
			\path[draw] (s2)--(s3);
			\path[draw] (s3)--(s1);
			\path[draw] (s4)--(s3);		
	
			\begin{scope}[dashed]
				\path[draw] (s1)--(s2);
			\end{scope}
	
			\path[draw] (s234)--(s2);
			\path[draw] (s234)--(s3);
			\path[draw] (s234)--(s4);

	\end{scope}
\end{tikzpicture}
\caption{A tetrahedron after one stacking operation.}
\label{fig:stackedTruncated:stacked}
\end{subfigure}
\hskip3ex
\begin{subfigure}[t]{0.4\linewidth}
\centering
\begin{tikzpicture}[line width =1pt]

	\begin{scope}[yshift=1cm, x={(1cm,0cm)}, y={(0.35cm,0.20cm)}, z={(0cm,1cm)}, scale=0.5]

				\path (4,2,1) node[name=t1, shape=coordinate]{};
				\path (-4,2,1) node[name=t2, shape=coordinate]{};
				\path (0,-4,1) node[name=t3, shape=coordinate]{};
				\path (0,0,-3) node[name=t4, shape=coordinate]{};
				
				\path ($3/4*(t1)+1/4*(t2)$) node[name=t12, shape=coordinate]{};

				\path ($1/3*(t1)+2/3*(t3)$) node[name=t13, shape=coordinate]{};
				
				\path ($3/4*(t1)+1/4*(t4)$) node[name=t14, shape=coordinate]{};

				\path[draw] (t2)--(t3);
				\path[draw] (t3)--(t4);
				\path[draw] (t4)--(t2);
				\path[draw] (t2)--(t12);
				\path[draw] (t3)--(t13);
				\path[draw] (t4)--(t14);

				\path[draw] (t12)--(t13)--(t14)--cycle;
		
	\end{scope}
\end{tikzpicture}
\caption{A tetrahedron after the corresponding truncation.}
\label{fig:stackedTruncated:truncated}
\end{subfigure}

\caption{} 
\label{fig:stackedTruncated}
\end{figure}

\subsection{Our results}
In this paper we introduce a duality transform that maintains a polynomial grid size. 
In other words, we provide a technique that takes a
grid embedding of a simplicial polytope with graph $G$ 
and generates a grid embedding of a polytope whose graph is $G^*$, the dual graph of $G$. 
We call a 3d polytope with graph $G^*$ a \emph{dual polytope.}

We prove the following result:
\begin{theorem*}[Theorem 10]
    Let $G$ be a triangulation with maximal vertex degree $\maxdeg$
    and let  $\P=(u_i)_{1\le i\le n}$ be a realization of $G$ as a convex polytope with integer coordinates.
    Then there exists  a realization $(\phi_f)_{f\in F(G)}$ of
    the dual graph $G^*$ as a convex polytope with integer coordinates bounded by\footnote{
For convenience, throughout the paper we use $|u|$ for the Euclidean norm of the vector $u$.
    }
    $$
    |\phi_f|<O(n\max|u_i|^{3\maxdeg + 9}).
    $$
\end{theorem*}

This, in particular, implies, that 
if we only consider simplicial polytopes with bounded vertex degrees
and with integer coordinates bounded by a polynomial in $n$, then the dual 
polytope obtained with our techniques 
has also integer coordinates bounded by a (different) polynomial in $n$.
Although our bound is not purely polynomial, it is in general an improvement over
the standard approaches for  constructing dual polytopes; see Sect.~\ref{subsec:dual}.

 
 For the class of stacked polytopes (although their maximum vertex degree is not bounded)
 we can also apply our approach
  to show that all graphs
 dual to planar 3-trees can be embedded as polytopes on a polynomial size grid. These polytopes
 are known as truncated polytopes. Truncated 3d polytopes are simple polytopes, which
 can be generated from a tetrahedron and a sequence of \emph{vertex truncations}.
 A vertex  truncation is the dual operation to stacking (Fig.~\ref{fig:stackedTruncated}). This means that a degree-3 vertex
 of the polytope is cut off by adding a new bounding hyperplane that separates this vertex
 from the remaining vertices of the polytope. We prove the following theorem.
 
 \begin{theorem*}[Theorem 4]
    Any truncated 3d polytope  with $n$ vertices can be realized
    with integer coordinates of size $O(n^{9\log6+1})$.
 \end{theorem*}

\setcounter{theorem}{1}
\addtocounter{theorem}{-1}

The proof of Theorem~\ref{thm:truncated} uses the strong
available results on planar realizations 
of graphs of the dual (stacked) polytopes (see~\cite{DS11}).

Such results are not available for general simplicial polytopes,
though we make a small step further in Theorem~\ref{thm:reverseMC}.
To prove the general Theorem~\ref{thm:simplicialDuality} we develop
some novel techniques to work with equilibrium stresses
directly in $\R^3$.


\subsection{Duality} 
\label{subsec:dual}
There exist several natural approaches how to construct a dual polytope.
To the best of our knowledge, all of them increase
the coordinates of the original polytope in general by an exponential factor when scaled to integers.

The most prominent construction is polarity with respect to the sphere. Let $P$ be a polytope
that contains the origin. Then $P^* = \{y \in \R^d\colon  \langle x , y\rangle \le 1 \text{ for all } x \in P \}$
is a polytope dual to $P$, called its \emph{polar}. The vertices of $P^*$ 
are intersection points of planes with integral normal vectors, and hence not necessarily
integer points. In order to scale to integrality one has to multiply $P^*$ with the product
of all denominators of its vertex coordinates, which may cause an exponential increase of
 the grid size.
 
An alternative approach goes via polarity with respect to the paraboloid: Every supporting hyperplane of a polytope facet can 
be described as the set of points $(x,y,z)^T$, for which 
the equation $ax+by = z+c$ holds 
($a$, $b$, $c$ are parameters which depend on the hyperplane, 
the construction is not applicable to hyperplanes parallel to the $z$-axis). 
By mapping each facet
 to a vertex of the dual polytope with coordinates $(a,b,c)^T$ we obtain the desired polarity transformation. 
This transformation can also be formulated in terms of
reciprocal diagrams and the Maxwell--Cremona correspondence~\cite{M64}.

Polarity with respect to the paraboloid does not necessarily provide small integer coordinates for two reasons.
First, for a facet whose boundary points have integer coordinates, 
the parameters $a,b,c$ of the supporting hyperplane given by $ax+by = z+c$ are rational.
Thus, the polar polytope is realized with rational coordinates and an exponential factor might be necessary
when scaling to integers.
Second, the construction realizes the dual polytope in the projective space with one point
``over the horizon''. The second property can be ``fixed'' with a projective transformation. This, however, 
makes a large scaling factor for an integral embedding unavoidable in the general case.

\subsection{Structure}

All of our algorithms work in two stages.
We first construct an intermediate object that we call a
{\em cone-convex embedding} of the graph. This object will be equipped with special edge weights, which we store
in a matrix that is called {\em CDV matrix}.
In the second stage we use an adaptation of the duality transform as introduced by Lov{\'a}sz 
(see~\cite{L01}) to transform a cone-convex embedding of the primal graph to an embedding of the dual as a convex polytope.

We present a duality transform for stacked polytopes (Sect.~\ref{sec:truncated}), simplicial polytopes with a degree~3 vertex (Sect.~\ref{sec:WDreverseMC}), and 
for general simplicial polytopes (Sect.~\ref{sec:simplicialDuality}).
The second stage is always carried out in the same way for all our algorithms.
It is therefore presented first in Sect.~\ref{sec:lovasz}.
We conclude our presentation in Sect.~\ref{sec:example} where we present an example of the embedding algorithm 
for truncated polytopes (Theorem~\ref{thm:truncated}).

\subsection{Notation and Conventions} 
 
We denote by $G$ the graph of the original polytope, and by $G^*$
its dual graph. For any graph $H$ we write $V(H)$ for its vertex set, $E(H)$ for its edge set and
$N(H, v)$ for the set of
neighbors of a vertex $v$ in $H$.
Since we consider 3-connected planar graphs, the facial structure of the graph is predetermined
up to a global reflection~\cite[Theorem 11]{W33}. The set of faces is therefore predetermined, and we name it $F(H)$.
%
%
We denote the maximum vertex degree of a graph $G$ as $\maxdeg$. Finally, we write $G[X]$ for the induced
subgraph of a vertex set $X\subseteq V(G)$.

Every embedding of a graph to $\R^d$ is always understood as 
a straight-line embedding.
Thus, an embedding of a graph to $\R^d$ is defined by a map $\mathbf p: V\to \R^d$
that assigns coordinates to every vertex.
For simplicity we denote $\mathbf p(v_i)$ with $p_i$ for a graph with the vertex set $V=(v_i)_{1\le i\le n}$.
Naturally, an embedding $\mathbf p$ of a graph with the vertex set $V=(v_i)_{1\le i\le n}$ to $\R^d$
can be seen as a point in $\R^{nd}$, that is $\mathbf p = (p_i)_{1\le i\le n}\in \R^{nd}$.
We typically use the letter $v$ to denote a vertex of an abstract graph, 
the letter $p$ for a vertex embedded in $R^2$,
and the letter $u$ for a vertex embedded in $\R^3$.

For convenience we use the \emph{square bracket notation} for oriented volumes and areas (for 2d vectors $[p_ip_jp_k]$ is defined similarly):
$$[u_iu_ju_ku_l]:=\det\begin{pmatrix}
x_i &x_j &x_k &x_l\\
y_i &y_j &y_k &y_l\\
z_i &z_j &z_k &z_l\\
 1 &1 &1 &1
\end{pmatrix},\quad
\text{ where }u=\begin{pmatrix}
x\\y\\z
\end{pmatrix}\in \R^3.
$$

Unless explicitly stated otherwise, the graphs we consider are planar and 3-connected. 
We call an embedding of a graph into $\R^d$ an {\em integer embedding} if the coordinates of its vertices are all integers.
If not explicitly stated otherwise, the surfaces of polytopes are oriented
so that the normal vector is an exterior normal vector to this polytope.

\setcounter{theorem}{1}
\addtocounter{theorem}{-1}
\section{Lov{\'a}sz' Duality Transform}
\label{sec:lovasz}

In this section we review some of the methods Lov{\'a}sz introduced in his paper on Steinitz representations \cite{L01}.


\begin{definition}
We call an embedding $\mathbf u = (u_i)_{i\le 1\le n}$ of a planar
3-connected graph $G$ in $\mathbb{R}^3$ a \emph{cone-convex
embedding}, if its projection onto the sphere $\left(\frac{u_i}{|u_i|}\right)_{1\le i\le n}$ 
with edges drawn as geodesic arcs
is a strictly convex embedding of $G$ into the sphere.
\end{definition}

We remark that the embedding of a graph $G$ into the sphere is strictly convex
if the faces of the embedding are 
strictly convex spherical polygons, 
or, in other words, the faces are the intersections of
pointed convex disjoint polyhedral cones with the sphere.
So, an embedding is cone-convex if  the
cones over its faces are pointed, convex and disjoint.
%
%
%
%
Note that the vertices of a cone-convex embedding are not supposed to form a convex polytope.

\begin{definition}
 Let $\mathbf u = (u_i)_{1\le i\le n}$ be an embedding
 of a graph $G$ into $\mathbb{R}^d$. 
 We call a symmetric matrix $M=[M_{ij}]_{1\le i,j\le n}$ a \emph{CDV matrix} of the embedding if
 \begin{enumerate}
 \item $M_{ij}=0$, for $i\ne j, (v_iv_j)\not \in E(G)$, and
 \item  $\sum_{1\le j\le n}M_{ij}u_j=0$, for $1\le i\le n$.
 \end{enumerate}
 We call a CDV matrix \emph{positive} if $M_{ij}>0$ for all $(v_iv_j)\in
E(G)$. 
\end{definition}
We remark that for a positive CDV matrix there are no conditions 
on the signs of the diagonal elements $M_{ii}$.
Moreover, in our applications the diagonal elements
play a purely technical role and never appear in geometrical constructions. Thus, we never bound the diagonal elements
and do not put conditions on their signs.
Note also that the CDV matrices of an embedding form a linear space
with respect to the standard matrix addition and multiplication with scalars. 

We refer to  the second condition in the above definition as the \emph{CDV equilibrium condition}. 
The CDV equilibrium condition can also be expressed in a slightly different, more
geometric form as  
 \begin{equation}\label{eq:CDV}
 \sum_{v_j\in N(G,v_i)}M_{ij}u_j=-M_{ii}u_i,\quad 1\le i\le  n.
 \end{equation}
%

The name CDV matrix was chosen in correspondence with 
Colin de Verdi{\`e}re matrices as defined in~\cite{L01}.
%
In our paper, however,  we use only the geometrical arguments from~\cite{L01} 
and thus do not rely on the more restrictive and more technical notion of Colin de Verdi{\`e}re matrices, 
which has its roots in spectral graph theory.
We remark, that by~\cite[Theorem~7]{L01} and the following discussion,
every positive CDV matrix of a cone-convex embedding of a graph $G$ is
a Colin de Verdi{\`e}re matrix of $G$, whose entries are multiplied with~-1.
%
%

Since the CDV matrix is a natural 3d counterpart to the 2d notion of equilibrium stress 
(which we introduce in Sect.~\ref{subsec:truncated:basicEquiCDV})
we refer to its entries as stresses.
We show the connection between these two notions in Lemma~\ref{lm:equiCDV}
and in Subsect.~\ref{subsec:simplicialDuality:equiCDVproj}.

The following lemma is due to Lov{\'a}sz~\cite{L01}, 
we include the proof since it illustrates 
how to construct a realization out of a CDV matrix.
\begin{lemma}[Lemma 4, \cite{L01}]\label{lem:LoLm1}
 Let $\mathbf u = (u_i)_{1\le i\le n}$ be a cone-convex embedding of a planar 3-connected graph $G$ with a positive
 CDV matrix  $M$. Then every face $f$ in $G$ can be assigned with a 
 vector $\phi_f$, s.t. for each adjacent face $g$ and separating edge $(v_iv_j)$
\begin{equation} \label{eq:phi}
\phi_f-\phi_g=M_{ij}(u_i\times u_j),
\end{equation}
where $f$ lies to the left and $g$ lies to the right from $\overrightarrow{u_iu_j}$.
The set of vectors $(\phi_f)$ is uniquely defined  up to
translations.
\end{lemma}
%
In the remaining of the paper we refer to a pair of dual edges 
$(u_iu_j)$ and $(\phi_f\phi_g)$ from the above lemma 
without explicitly mentioning their orientation.
\begin{proof}
 To construct the family of vectors $(\phi_f)$, we start by assigning an
arbitrary value to $\phi_{f_0}$ (for an arbitrary face $f_0$); then we proceed
 iteratively picking $\phi$ vectors such that Eq.~\eqref{eq:phi} is fulfilled. To prove the consistency of the construction, we
show that the differences $(\phi_f-\phi_g)$
 sum to zero over every cycle in $G^*$. Since $G$
as well as $G^*$ is planar and 3-connected, it suffices to check this condition for
all elementary cycles of $G^*$, which are the faces of $G^*$. 
Let $\tau(i)$ denote the set of counterclockwise oriented edges of the face in $G^*$ dual to $v_i\in V(G)$. Then, combining~\eqref{eq:CDV} and~\eqref{eq:phi} yields
\begin{align*}
 \sum_{(f,g)\in \tau(i)} \ns (\phi_{f}-\phi_{g})   
 =\ns \sum_{v_j\in N(G,v_i)} \ns M_{ij}(u_i\times u_j)
 &=u_i\times \left(\sum_{v_j\in N(G,v_i)} \ns M_{ij}u_j\right)\\
 &=u_i\times (-M_{ii}u_i)=0.
\end{align*}
The vectors $\phi_f$ are unique up to the initial choice for $\phi_{f_0}$.
\end{proof}

Note that there is a canonical way to derive a CDV matrix from a 3d polytope~\cite{L01}.
Every 3d embedding $\mathbf u$ of a graph $G$ as a polytope possesses a 
CDV matrix defined by the vertices $(\phi_f)_{f\in F(G)}$ of its polar and Eq.~\eqref{eq:phi}.
We call theses matrices {\em canonical CDV matrices}.
For a convex polytope containing the origin in its interior 
the canonical CDV matrix is positive.
Canonical CDV matrices are used in 
Sect.~\ref{sec:simplicialDuality} where we give 
the full definition and details.

The vectors constructed in Lemma~\ref{lem:LoLm1} satisfy the following crucial property:
\begin{lemma} [Lemma 5, \cite{L01}] \label{lem:LoLm2}
 Let $\mathbf u = (u_i)_{1\le i\le n}$ be a cone-convex embedding of a planar 3-connected graph $G$ with a positive
 CDV matrix  $M$. Then for any set of vectors $(\phi_f)_{f\in F(G)}$ fulfilling~\eqref{eq:phi},
the convex hull $\Conv((\phi_f)_{f\in F(G)})$ 
is a convex polytope with graph $G^*$; and the
isomorphism between $G^*$ and the skeleton of $\Conv((\phi_f)_{f\in F(G)})$ is given by $f\to \phi_f$.
\end{lemma}

Lov{\'a}sz formulates this result with an additional
 constraint $|u_i|=1$ for all vertices of $G$, that is for an embedding of a graph onto the sphere.
However, he never uses this constraint in the proof.
Additionally, he formulates the rescaling property of CDV matrices
showing why the renormalization of $\mathbf u$ doesn't change the result
(we review this property in details later in 
Lemma~\ref{lm:CDVscalable}).

Lemma~\ref{lem:LoLm1} and~\ref{lem:LoLm2} 
imply the following theorem which is a main tool in later constructions:
\begin{theorem}\label{thm:LoThm}
    Let $G$ be a planar 3-connected graph, 
    let $\mathbf u = (u_i)_{1\le i\le n}$ be its cone-convex embedding with integer coordinates
    and let $M$ be a positive CDV matrix with integer entries.
    Then there exists  a realization $(\phi_f)_{f\in F(G)}$ of
    the dual graph $G^*$ as a convex polytope with integer coordinates bounded by

    $$
    |\phi_f|<2n\cdot \max_{(v_i,v_j) \in E(G)}|M_{ij}(u_i\times u_j)|.
    $$
\end{theorem}

\begin{proof}
We use Lemma~\ref{lem:LoLm1} to construct $(\phi_f)_{f\in F(G)}$ that satisfy~\eqref{eq:phi}
and such that $\phi_{f_0}=(0,0,0)^T$ for a distinguished face $f_0\in F(G)$. Lemma~\ref{lem:LoLm2}
guaranties that $(\phi_f)_{f\in F(G)}$ form a convex polytope with graph $G^*$. 
Since $(\phi_f)$ satisfy~\eqref{eq:phi}, $\phi_{f_0}=(0,0,0)^T$ and 
all $M_{ij}$ as well as all $u_i$ are integral, all $\phi_f$ have integer coordinates as well.

To finish the proof we estimate how large the vectors $(\phi_f)$ are.
We evaluate $\phi_{f_k}$ for some face $f_k\in F(G)$. The following algebraic 
expression holds for all values $\phi_{f_i}$:
 $$
 \phi_{f_k}=\phi_{f_0}+(\phi_{f_1}-\phi_{f_0})+\ldots +(\phi_{f_{k-1}}-\phi_{f_{k-2}})+(\phi_{f_k}-\phi_{f_{k-1}}).
 $$
 Let us now consider 
the shortest path $f_0,f_1,\ldots, f_k$  in $G^*$ connecting the faces $f_0$
and $f_k$. Clearly, $k$ is less than $2n$, and hence
 \begin{align*}
 |\phi_{f_k}|
 &\le 2n\cdot \max_{(f_a,f_b)\in E(G^*)}|\phi_{f_a}-\phi_{f_b}|
 =2n\cdot \max_{(v_i,v_j) \in E(G)}|M_{ij}(u_i\times u_j)|.
 \end{align*}
\end{proof}

\section{A Duality Transform for Truncated polytopes}
\label{sec:truncated}

\subsection{Equilibrium stresses} 
\label{subsec:truncated:basicEquiCDV}
We describe next  the connection between 
convex 2d embeddings with positive equilibrium stresses and 
cone-convex 3d embeddings with positive CDV matrices.
We follow the presentation of~\cite{DS11} and define:
\begin{definition}
An assignment $\omega\colon E(G)\to\mathbb{R}$ of scalars
(denoted by $\omega(i,j)=\omega_{ij}=\omega_{ji}$)
to the edges of a graph $G$ is called a \emph{stress}.
A stress is an \emph{equilibrium stress} for
an embedding $\mathbf u = (u_i)$ of $G$ into $\mathbb{R}^d$ if
for every vertex $v_i\in V(G)$
 $$
 \sum_{v_j\in N(G,v_i)} \omega_{ij}(u_j-u_i)=0.
 $$
We call an equilibrium stress of a 2d embedding with a distinguished boundary face $f_0$
\emph{positive} if it is positive on every edge that does not belong to $f_0$.
 \end{definition}

The concept of equilibrium stress plays a central role in the classical Maxwell--Cremona lifting
approach and it is also a crucial concept in our embedding algorithm.
The following lemma establishes some preliminary connections 
between equilibrium stresses and CDV matrices.


\begin{lemma}\label{lm:equiCDV}
 Let $\mathbf u = (u_i)_{1\le i\le n}$ be an embedding of a graph $G$ into $\mathbb{R}^3$. The following three statements hold:
 \begin{enumerate}
  \item\label{lm:equiCDV:1} Let $G^+$ be the graph $G$ with one additional vertex $v_0$ 
connected with every vertex in $G$,
and let $\mathbf{u^+}=(u_0=(0,0,0)^T,u_1,\ldots,u_n)$ be 
an  embedding of $G^+$ into $\mathbb{R}^3$ equipped with 
an equilibrium stress $\omega$. 
Then the assignment
          $$M_{ij}:=
	\begin{cases}
 		-\sum_{v_k\in N(G^+,v_i)}\omega_{ik}, \quad &1\le i=j\le n,\\
	  	\omega_{ij},				&(i,j)\in E(G),\\
		0, 					&\text{else}
	\end{cases}
	$$
defines a CDV matrix $M$ for the embedding $\mathbf u$ of $G$.


  \item\label{lm:equiCDV:2}  Let $\omega$ be an equilibrium stress for the embedding $\mathbf u$. Then the assignment
    $$
	M_{ij}:=\begin{cases}
	    -\sum_{v_k\in N(G,v_i)}\omega_{ik},  	&1\le i=j\le n, \\
	    \omega_{ij},				&(i,j)\in E(G),\\
	    0, 					&\text{else}
	\end{cases}
    $$
    defines a CDV matrix $M$ for $\mathbf u$.
  \item\label{lm:equiCDV:3}  Let $\mathbf u$ be a flat embedding that lies in a plane 
not containing the origin. Let $M$ be a CDV matrix for $\mathbf u$. 
Then $(M_{ij})_{(i,j)\in E(G)}$ defines an equilibrium stress for $\mathbf u$.
 \end{enumerate}
\end{lemma}

\begin{proof}
\ref{lm:equiCDV:1}. We check that the CDV equilibrium holds by noting that for every $i$:
\begin{align*}
\quad \sum_{1\le j\le n}M_{ij}u_j 
= \sum_{v_j\in N(G,v_i)}M_{ij}u_j+M_{ii}u_i 
&= \!\!\!\!\! \sum_{v_j\in N(G,v_i)} \!\!\!\!\!\omega_{ij}u_j-(\sum_{v_j\in N(G^+,v_i)} \!\!\!\!\! \omega_{ij})u_i\\
&= \!\!\!\!\! \sum_{v_j\in N(G^+,v_i)} \!\!\!\!\!\omega_{ij}(u_j-u_i)+\omega_{i0}u_0=0.
\end{align*}
The second transition holds due to the definition of $M_{ij}$ and $M_{ii}$.
The last transition holds since $\sum_{v_j\in N(G^+,v_i)}\omega_{ij}(u_j-u_i)=0$ by the
definition of equilibrium stress and $u_0=0$.

\ref{lm:equiCDV:2}. We construct the embedding $(u_0=(0,0,0)^T,u_1,\ldots,u_n)$ 
of the graph $G^+:=G+\{v_0\}$
and extend the equilibrium stress $\omega$ of $G$ 
to an equilibrium stress $\omega^+$ of $G^+$
by assigning zeros to all the new edges $\omega^+_{i0}:=0$. 
Then we use part~\ref{lm:equiCDV:1} of the lemma to finish the proof.

\ref{lm:equiCDV:3}. We denote by $\alpha$ the plane of the flat embedding $\mathbf u$. We rewrite the CDV equilibrium condition:
$$
0= \sum_{1\le j \le n}M_{ij}u_j = 
\sum_{v_j\in N(G,v_i)}\ns M_{ij}(u_j-u_i)
    +\left(M_{ii}+\ns \sum_{v_j\in N(G,v_i)} \ns M_{ij}\right)u_i
$$
and notice that, the first summand, if nonzero, is parallel to the plane $\alpha$ while the second is not, so both must equal zero and thus
$$
\sum_{v_j\in N(G,v_i)}\ns M_{ij}(u_j-u_i)=0.
$$
\end{proof}

\subsection{The stacking approach} 
\label{subsec:truncated:framework}

Here we present an approach for constructing 3d cone-convex embeddings
with CDV matrices from 2d convex embeddings with equilibrium stress
by stacking an additional vertex (tetrahedron) to the graph.
%
%
We denote a graph obtained from $G$ by stacking a vertex $v_1$ on a face $(v_2v_3v_4)$ by $\Stack(G;v_1;v_2v_3v_4)$.

\begin{theorem}\label{thm:framework}
Let $\mathbf p = (p_i)_{2\le i\le n}$ be a 2d integer planar convex embedding of a planar 3-connected graph $G^{\uparrow}$
with a designated triangular face $(v_2v_3v_4)$ embedded as the
boundary face. 
Let $\omega$ be a positive integer equilibrium stress for $\mathbf p$.

Then the embedding $\mathbf u = (u_i)_{1\le i\le n}$ of the graph  
$G=\Stack(G_{\uparrow};v_1;v_2v_3v_4)$ defined as
\begin{align*}
u_i &= (3p_i-(p_2+p_3+p_4),1)^T, \quad 2\le i\le n, \\
u_1 &= (0,0,-3)^{T},
\end{align*}
where $(p,1)^T  = (p_x, p_y, 1)^{T} \in \R^3$ for $p=(p_x, p_y)\in \R^2$,
is an integral cone-convex embedding 
and there exists a positive integer CDV matrix $M$ 
for $\mathbf u$ such that
\begin{align*}
M_{ij} &= \omega_{ij}, \quad\text{for each internal edge $(i,j)$
    of the original embedding of }G_{\uparrow},\\
|M_{ij}| &\le \max_{(i,j)\in E(G^{\uparrow})}|\omega_{ij}| + 1, \qquad \forall (i,j)\in E(G).
\end{align*}
%
\end{theorem}

\begin{proof}

The embedding $(u_i)_{1\le i\le n}$ can be described as follows:
The embedding of $G_{\uparrow}$ 
is realized in the plane $\{z=1\}$,
scaled 3 times,
and translated so that the barycenter of the boundary face coincides with the origin.
The stacked vertex is then placed at $(0,0,-3)^T$.
The embedding is cone-convex since it describes a tetrahedron containing the origin with one face
that is refined with a plane convex subdivision.

Following the structure of $G=\Stack(G_{\uparrow};v_1;v_2v_3v_4)$, we
decompose $G$ into two subgraphs: $G_{\uparrow}=G[\{v_2,\ldots,v_n\}]$
and $G_{\downarrow}:=G[\{v_1,v_2,v_3,v_4\}]$.

We first compute a CDV matrix $[M'_{ij}]_{2\le i,j\le n}$ for the
embedding $(u_i)_{2\le i\le n}$ of $G_{\uparrow}$. 
The plane embedding $(p_i)_{2\le i\le n}$ of  $G_{\uparrow}$ comes with an integer equilibrium stress $\omega$.
Since $(u_i)_{2\le i\le n}$ is just 
a rescaling and translation of $(p_i)_{2\le i\le n}$, clearly, 
$\omega$ is also an equilibrium stress for 
$(u_i)_{2\le i\le n}$ and we use part~\ref{lm:equiCDV:2} of Lemma~\ref{lm:equiCDV}
to transform it into the integer CDV matrix $[M'_{ij}]_{2\le i,j\le n}$.
	
As a second step we compute a CDV matrix $[M''_{ij}]_{1\le i,j\le 4}$ for the embedding
 of the tetrahedron $G_{\downarrow}$. 
The tetrahedron $G_{\downarrow}$ possesses a trivial CDV matrix:
By the construction, $u_1+u_2+u_3+u_4 = 0$. Thus, the matrix
$$
M'':=\begin{pmatrix}
1 &1 &1 &1\\
1 &1 &1 &1\\
1 &1 &1 &1\\
1 &1 &1 &1
\end{pmatrix}
$$
is a CDV matrix for the embedding $(u_1, u_2, u_3, u_4)$ 
of $G^{\uparrow}$.

	
In the final step we extend the two CDV matrices $M'$ and $M''$ to $G$
and combine them. Clearly, a CDV matrix padded with zeros remains a CDV matrix. Furthermore,
any linear combination of CDV matrices is again a CDV matrix.
Thus, we form an integer
CDV matrix for the whole embedding $(u_i)_{1\le i\le n}$ of $G$ by setting:
	$$
	M:=M'+\lambda M'',
	$$
where $\lambda$ is a positive integer chosen so that $M$ is a positive CDV matrix. This can be done as follows.
Recall that  $\omega$ is a positive stress and $M''$ is a positive CDV matrix. Hence, 
 the only six entries in $M$ corresponding to edges of $G$ that may be negative
 are: $M_{23}$, $M_{34}$ and $M_{42}$ (and their symmetric entries), for which 
	$M_{ij}:=M'_{ij}+\lambda M''_{ij}$ with $M'_{ij}=\omega_{ij}<0$ and $M''_{ij}=1$.
	Thus, we choose $\lambda$ such that $M$ is positive at these entries.
	To satisfy this condition we pick
	$$
	    \lambda:=\max_{(i,j)\in \{(2,3),(3,4),(4,2)\}}|M'_{ij}|+1.
	$$
	The bound $|M_{ij}| \le \max_{kl}|\omega_{kl}|+1$ for every edge $(i,j)$ of $G$ trivially follows.
\end{proof}
    
\subsection{Realizations of Truncated Polytopes}

\label{subsec:truncated:truncated}

We can now combine the previous results in Theorem~\ref{thm:frameworkL} before presenting the embedding algorithm for truncated
polytopes in Theorem~\ref{thm:truncated}.
\begin{theorem}\label{thm:frameworkL}
 Let $G=\Stack(G_{\uparrow};v_1;v_2v_3v_4)$ and let $\mathbf p = (p_i)_{2\le i\le n}$
 be a planar 2d embedding of $G_{\uparrow}$
 with integer coordinates, boundary face $(v_2v_3v_4)$, and
 a positive integer equilibrium stress $\omega$.
 Then there exists a realization $(\phi_f)_{f\in F(G^*)}$ of 
 the  graph $G^*$, dual to $G$, as a convex polytope with integer coordinates such that  
 $$
 |\phi_f| = O(n\cdot \max|\omega_{ij}| \cdot\max|p_i|^2).
 $$
\end{theorem}
\begin{proof}
We first apply Theorem~\ref{thm:framework} 
to obtain a cone-convex embedding $\mathbf u = (u_i)_{1\le i\le n}$ of $G$
with integer coordinates and
a positive integer CDV matrix $M$.
We then apply  Theorem~\ref{thm:LoThm}
and obtain a family of vectors $(\phi_f)_{f\in F(G^*)}$ 
that forms a desired realization of $G^*$ as a convex polytope with integer coordinates.

To  estimate how large the coordinates of the embedding $(\phi_f)$ are,
we combine bounds for $\phi_f$ given by Theorem~\ref{thm:LoThm} with the bounds for the 
entries of $M$ given by Theorem~\ref{thm:framework}:
 \begin{align*}
 |\phi_f|
 \le 2n\cdot \max_{(v_i,v_j) \in E(G)}|M_{ij}(u_i\times u_j)|
 &\le 2n\cdot (\max|\omega_{ij}|+1) \cdot\max|u_i|^2\\
 &=O(n\cdot \max|\omega_{ij}| \cdot\max|p_i|^2).
 \end{align*}
\end{proof}

Next we apply Theorem~\ref{thm:frameworkL} to construct an
integer polynomial size grid embedding for truncated polytopes.
To construct small integer 2d embeddings with a small integer equilibrium stress
we use a result by Demaine~and~Schulz~\cite{DS11}:
\begin{lemma}\label{lm:ds}
Any graph of a
stacked polytope with $n$ vertices and any distinguished face $f_0$ 
can be embedded on a $O(n^{2\log6}) \times O(n^{2\log6})$ grid 
with boundary face $f_0$ and with integral positive equilibrium stress $\omega$
such that, for every edge $(i,j)$, we have $|\omega_{ij}|=O(n^{5\log6})$.
\end{lemma}
\begin{proof}
This lemma is not explicitly stated in~\cite{DS11} but it is a
by-product of the presented constructions.
In particular, the estimates for the coordinate size are stated there in
Theorem~1 and the integer stresses are defined there by Lemma 9.

Recently, a flaw was discovered in~\cite{DS11}, which affected the estimates.
The flaw was within the balancing procedure, Sect.~3.1 of~\cite{DS11}.
It was corrected in the latest version of the paper~\cite[Lemma~7]{DS16}. 

The fix affected the estimate of Eq.~(3.3) in~\cite{DS11} and the numbers in the following computations (but not the arguments). Since we need
the corrected computations for the bounds of this lemma 
we report briefly the necessary changes using the exact notations of the original paper.
The updated bound of Eq.~(3.3) is $w(f_0)\le n^{\log6}$~\cite[Lemma~7]{DS16} instead of $n^2$. We denote this bound with $R$.
To correct the values in Sect.~3.3
we replace in the computations $n$ with $\sqrt R = n^{(\log6)/2}$.
In particular, in Subsect.~3.3.1 the (preliminary) boundary face is now given by 
$\mathbf p_1 = (0,0), \mathbf p_2 = (\sqrt R,0), \mathbf p_3 = (0, \sqrt R)$.
In Subsect.~3.3.2 the (final) boundary face is now given by 
$\mathbf p_1 = (0,0), \mathbf p_2 = (10 R^2,0), \mathbf p_3 = (0, 10 R^2)$;
and the corrected scaling factor for the stress is $Y = 4 R $.
The lemma follows by taking the updated values from~\cite{DS11}.
\end{proof}
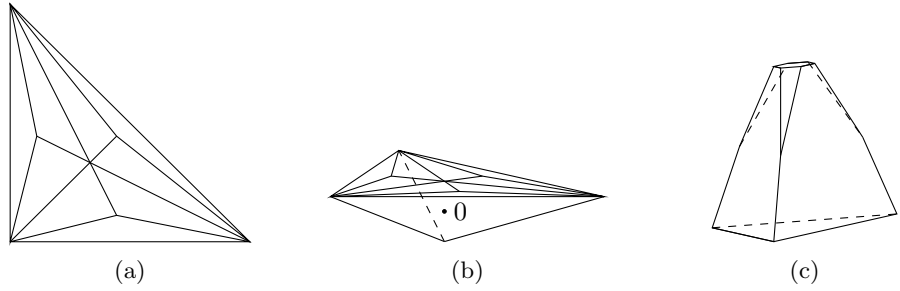
\begin{figure}[t]
\centering
\begin{subfigure}[t]{0.3\linewidth}
\centering	
\begin{tikzpicture}

	\begin{scope}[scale=0.35, x={(1cm,0cm)}, y={(0cm,1cm)}]
		\path (-3,-3) node[name=p2, shape=coordinate]{};
		\path (6,-3) node[name=p3, shape=coordinate]{};
		\path (-3,6) node[name=p4, shape=coordinate]{};
		\path (0,0) node[name=p5, shape=coordinate]{};
		\path (1,-2) node[name=p6, shape=coordinate]{};
		\path (1,1) node[name=p7, shape=coordinate]{};
		\path (-2,1) node[name=p8, shape=coordinate]{};

		\path[draw]  (p2) -- (p3) -- (p4) -- (p2) -- (p8) -- (p4) -- (p7) -- (p3) -- (p6) -- (p2);
		\path[draw]  (p3) -- (p5) -- (p4) ;
		\path[draw] (p2) -- (p5) -- (p7);
		\path[draw] (p6) -- (p5) -- (p8);
	\end{scope}

\end{tikzpicture}
\caption{}
\label{fig:onlyfig:1}
\end{subfigure}	
\begin{subfigure}[t]{0.3\linewidth}	
\centering
\begin{tikzpicture}
	
	\begin{scope}[xshift=4.5cm, scale=0.4, x={(1cm,0cm)}, y={(0.25cm,0.17cm)}, z={(0cm,1cm)}]
		\path (-3,-3,1) node[name=q2, shape=coordinate]{};
		\path (6,-3,1) node[name=q3, shape=coordinate]{};
		\path (-3,6,1) node[name=q4, shape=coordinate]{};
		\path (0,0,1) node[name=q5, shape=coordinate]{};
		\path (1,-2,1) node[name=q6, shape=coordinate]{};
		\path (1,1,1) node[name=q7, shape=coordinate]{};
		\path (-2,1,1) node[name=q8, shape=coordinate]{};
		\path (0,0,-1) node[name=q1, shape=coordinate]{};
		
		\path[draw] (q2) -- (q3) -- (q4) -- (q2) -- (q8) -- (q4) -- (q7) -- (q3) -- (q6) -- (q2);
		\path[draw] (q3) -- (q5) -- (q4);
		\path[draw] (q2) -- (q5) -- (q7);
		\path[draw] (q6) -- (q5) -- (q8);
		\path[draw] (q1) -- (q2);
		\path[draw] (q1) -- (q3);
		\path[draw, dashed] (q1) -- (q4);
		
		\draw[fill] (0,0,0) ellipse(.06 cm and .06 cm);
		\node[right] at (0,0,0) {$0$};
	
	\end{scope}
\end{tikzpicture}
\caption{}
\label{fig:onlyfig:2}
\end{subfigure}	
\begin{subfigure}[t]{0.3\linewidth}	
\centering
\begin{tikzpicture}
	\begin{scope}[xshift=9cm, yshift=-0.7cm, scale=0.04, x={(1cm,0cm)}, y={(0.25cm,0.17cm)}, z={(0cm,1cm)}]
		\path (0,-18,27) node[name=f263, shape=coordinate]{};
		\path (-3,-6,54) node[name=f265, shape=coordinate]{};
		\path  (-6,-3,54) node[name=f258, shape=coordinate]{};
		\path  (-18,0,27) node[name=f284, shape=coordinate]{};
		\path  (-3,3,54) node[name=f485, shape=coordinate]{};
		\path  (3,6,54) node[name=f457, shape=coordinate]{};
		\path  (18,18,27) node[name=f473, shape=coordinate]{};
		\path  (6,3,54) node[name=f375, shape=coordinate]{};
		\path  (3,-3,54) node[name=f356, shape=coordinate]{};
		\path  (0,-18,27) node[name=f362, shape=coordinate]{};
		\path  (0,-27,0) node[name=f321, shape=coordinate]{};
		\path  (-27,0,0) node[name=f124, shape=coordinate]{};
		\path  (27,27,0) node[name=f143, shape=coordinate]{};
		
		\path[draw] (f362) -- (f265) -- (f258) -- (f284) -- (f124) -- (f321);
		\path[draw, dashed] (f284) -- (f485);
		\path[draw] (f485) -- (f457);
		\path[draw, dashed] (f457) -- (f473);
		\path[draw] (f473) -- (f143) -- (f321);
		\path[draw] (f473) -- (f375) -- (f356) -- (f362) -- (f321) -- (f124) ;
		\path[draw, dashed] (f124) -- (f143);
		\path[draw] (f375) -- (f356) -- (f265) -- (f258) -- (f485) -- (f457) -- (f375);
		
	\end{scope}
\end{tikzpicture}
\caption{}
\label{fig:onlyfig:3}
\end{subfigure}
\caption{The 2d embedding of $G_\uparrow$ (a), the cone-convex embedding of $G$ (b), 
and the resulting embedding of the dual (c).} 
\label{fig:onlyfig}
\end{figure}
\begin{theorem}
\label{thm:truncated}
Any truncated 3d polytope  with $n$ vertices can be realized
with integer coordinates of size $O(n^{9\log6+1})\subset O(n^{24.27})$.
\end{theorem}
\begin{proof}
Let $G^*$ be the graph of the truncated polytope and $G:=(G^{*})^{*}$ its dual.
Clearly, $G$ is the graph of a stacked polytope with $(n+4)/2$ vertices.
We denote the last stacking operation (for some 
sequence of stacking operations producing $G$) as the stacking of the vertex $v_1$ onto
the face $(v_2v_3v_4)$ of the  graph  $G_{\uparrow}:=G[V\setminus\{v_1\}]$. The 
graph $G_{\uparrow}$ is again a graph of a stacked polytope, and hence, 
by Lemma~\ref{lm:ds}, there 
exists an embedding $(p_i)_{2\le i\le n}$ of $G_{\uparrow}$ into $\Z^2$ with an equilibrium stress
$\omega$ satisfying the properties of Theorem~\ref{thm:frameworkL}.
We apply Theorem~\ref{thm:frameworkL} and obtain a polytope embedding $(\phi_f)$ of $G^*$ with bound
$$
|\phi_f| = O(n\cdot \max|\omega_{ij}| \cdot\max|p_i|^2) = O(n^{9\log6+1}).
$$
\end{proof}

Fig.~\ref{fig:onlyfig} shows an example of our algorithm. The computations for this example are presented in Sect.~\ref{sec:example}.

\section{A Duality Transform for Simplicial Polytopes with a Degree-3 Vertex}

In this section we develop a duality transform for
simplicial polytopes with a degree-3 vertex and some special geometric properties.
This result is an easy consequence of new results on 2d equilibrium stresses, which will be presented first:
the {\em wheel-decomposition} theorem and an efficient reverse of the Maxwell--Cremona lifting.

\label{sec:WDreverseMC}
%
%
%
%
%
%

\subsection{Wheel decomposition for equilibrium stresses}
\label{subsec:WDreverseMC:WD}
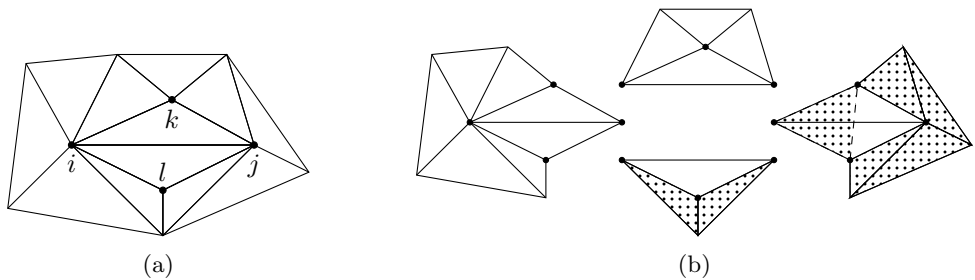
\begin{figure}[h]
\centering
\begin{subfigure}[t]{0.3\linewidth}
\centering
\begin{tikzpicture}[scale=1.2]

	\begin{scope}
		 
		\path (-1,0) node[name=p1, shape=coordinate]{};
		\path (1,0) node[name=p2, shape=coordinate]{};
		\path (0.1,0.5) node[name=p3, shape=coordinate]{};
		\path (0,-0.5) node[name=p4, shape=coordinate]{};
		\path (0,-1) node[name=p5, shape=coordinate]{};
		\path (-1.7,-0.7) node[name=p6, shape=coordinate]{};
		\path (-1.5,0.9) node[name=p7, shape=coordinate]{};
		\path (-0.5,1) node[name=p8, shape=coordinate]{};
		\path (0.7,1) node[name=p9, shape=coordinate]{};
		\path (1.6,-0.3) node[name=p10, shape=coordinate]{};

		\path[draw, fill] (p1) circle (1pt) node[anchor=north] {$i$};
		\path[draw, fill] (p2) circle(1pt) node[anchor=north] {$j$};
		\path[draw, fill] (p3) circle(1pt) node[anchor=north] {$k$};
		\path[draw, fill] (p4) circle(1pt) node[anchor=south] {$l$};

		\draw (p1)--(p2);
		\draw (p1)--(p3);
		\draw (p1)--(p4);
		\draw (p1)--(p5);
		\draw (p1)--(p6);
		\draw (p1)--(p7);
		\draw (p1)--(p8);
		\draw (p2)--(p4)--(p5)--(p6)--(p7)--(p8)--(p3)--cycle;

		\draw (p2)--(p4);
		\draw (p2)--(p1);
		\draw (p2)--(p3);
		\draw (p2)--(p9);
		\draw (p2)--(p10);
		\draw (p2)--(p5);
		\draw (p1)--(p3)--(p9)--(p10)--(p5)--(p4)--cycle;

		\draw (p3)--(p1);
		\draw (p3)--(p2);
		\draw (p3)--(p9);
		\draw (p3)--(p8);
		\draw (p1)--(p2)--(p9)--(p8)--cycle;

		\draw (p4)--(p1);
		\draw (p4)--(p2);
		\draw (p4)--(p5);
		\draw (p1)--(p2)--(p5)--cycle;

	\end{scope}
\end{tikzpicture}
\caption{}
\label{fig:wheelDecomposition:1}
\end{subfigure}
\hskip3ex
\begin{subfigure}[t]{0.6\linewidth}
\centering
\begin{tikzpicture}

	\begin{scope}[xshift=4cm]

		\path (-1,0) node[name=p1, shape=coordinate]{};
		\path (1,0) node[name=p2, shape=coordinate]{};
		\path (0.1,0.5) node[name=p3, shape=coordinate]{};
		\path (0,-0.5) node[name=p4, shape=coordinate]{};
		\path (0,-1) node[name=p5, shape=coordinate]{};
		\path (-1.7,-0.7) node[name=p6, shape=coordinate]{};
		\path (-1.5,0.9) node[name=p7, shape=coordinate]{};
		\path (-0.5,1) node[name=p8, shape=coordinate]{};
		\path (0.7,1) node[name=p9, shape=coordinate]{};
		\path (1.6,-0.3) node[name=p10, shape=coordinate]{};

		\path[draw, fill] (p1) circle(1pt) node[anchor=north] {};
		\path[draw, fill] (p2) circle(1pt) node[anchor=north] {};
		\path[draw, fill] (p3) circle(1pt) node[anchor=north] {};
		\path[draw, fill] (p4) circle(1pt) node[anchor=south] {};

		\draw (p1)--(p2);
		\draw (p1)--(p3);
		\draw (p1)--(p4);
		\draw (p1)--(p5);
		\draw (p1)--(p6);
		\draw (p1)--(p7);
		\draw (p1)--(p8);
		\draw (p2)--(p4)--(p5)--(p6)--(p7)--(p8)--(p3)--cycle;
	\end{scope}

	\begin{scope}[xshift=8cm]

		\path (-1,0) node[name=p1, shape=coordinate]{};
		\path (1,0) node[name=p2, shape=coordinate]{};
		\path (0.1,0.5) node[name=p3, shape=coordinate]{};
		\path (0,-0.5) node[name=p4, shape=coordinate]{};
		\path (0,-1) node[name=p5, shape=coordinate]{};
		\path (-1.7,-0.7) node[name=p6, shape=coordinate]{};
		\path (-1.5,0.9) node[name=p7, shape=coordinate]{};
		\path (-0.5,1) node[name=p8, shape=coordinate]{};
		\path (0.7,1) node[name=p9, shape=coordinate]{};
		\path (1.6,-0.3) node[name=p10, shape=coordinate]{};

		\path[draw, fill] (p1) circle(1pt) node[anchor=north] {};
		\path[draw, fill] (p2) circle(1pt) node[anchor=north] {};
		\path[draw, fill] (p3) circle(1pt) node[anchor=north] {};
		\path[draw, fill] (p4) circle(1pt) node[anchor=south] {};

		\draw (p2)--(p4);
		\draw (p2)--(p1);
		\draw (p2)--(p3);
		\draw (p2)--(p9);
		\draw (p2)--(p10);
		\draw (p2)--(p5);
		\draw (p1)--(p3)--(p9)--(p10)--(p5)--(p4)--cycle;

		\path[draw, pattern=dots] (p2)--(p3)--(p9)--cycle;
		\path[draw, pattern=dots] (p2)--(p9)--(p10)--cycle;
		\path[draw, pattern=dots] (p2)--(p10)--(p5)--cycle;
		\path[draw, pattern=dots] (p2)--(p5)--(p4)--cycle;
		\path[draw, dashed, pattern=dots] (p1)--(p3)--(p4)--cycle;

	\end{scope}

	\begin{scope}[xshift=6cm, yshift=0.5cm]

		\path (-1,0) node[name=p1, shape=coordinate]{};
		\path (1,0) node[name=p2, shape=coordinate]{};
		\path (0.1,0.5) node[name=p3, shape=coordinate]{};
		\path (0,-0.5) node[name=p4, shape=coordinate]{};
		\path (0,-1) node[name=p5, shape=coordinate]{};
		\path (-1.7,-0.7) node[name=p6, shape=coordinate]{};
		\path (-1.5,0.9) node[name=p7, shape=coordinate]{};
		\path (-0.5,1) node[name=p8, shape=coordinate]{};
		\path (0.7,1) node[name=p9, shape=coordinate]{};
		\path (1.6,-0.3) node[name=p10, shape=coordinate]{};

		\path[draw, fill] (p1) circle(1pt) node[anchor=north] {};
		\path[draw, fill] (p2) circle(1pt) node[anchor=north] {};
		\path[draw, fill] (p3) circle(1pt) node[anchor=north] {};

		\draw (p3)--(p1);
		\draw (p3)--(p2);
		\draw (p3)--(p9);
		\draw (p3)--(p8);
		\draw (p1)--(p2)--(p9)--(p8)--cycle;

	\end{scope}

	\begin{scope}[xshift=6cm, yshift=-0.5cm]

		\path (-1,0) node[name=p1, shape=coordinate]{};
		\path (1,0) node[name=p2, shape=coordinate]{};
		\path (0.1,0.5) node[name=p3, shape=coordinate]{};
		\path (0,-0.5) node[name=p4, shape=coordinate]{};
		\path (0,-1) node[name=p5, shape=coordinate]{};
		\path (-1.7,-0.7) node[name=p6, shape=coordinate]{};
		\path (-1.5,0.9) node[name=p7, shape=coordinate]{};
		\path (-0.5,1) node[name=p8, shape=coordinate]{};
		\path (0.7,1) node[name=p9, shape=coordinate]{};
		\path (1.6,-0.3) node[name=p10, shape=coordinate]{};

		\path[draw, fill] (p1) circle(1pt) node[anchor=north] {};
		\path[draw, fill] (p2) circle(1pt) node[anchor=north] {};
		\path[draw, fill] (p4) circle(1pt) node[anchor=south] {};

		\draw (p4)--(p1);
		\draw (p4)--(p2);
		\draw (p4)--(p5);
		\draw (p1)--(p2)--(p5)--cycle;

		\path[draw, pattern=dots] (p4)--(p1)--(p5)--cycle;
		\path[draw, pattern=dots] (p4)--(p2)--(p5)--cycle;

	\end{scope}

\end{tikzpicture}
\caption{}
\label{fig:wheelDecomposition:2}
\end{subfigure}
\caption{Part of a triangulation, participating in 
the wheel-decomposition of $\omega_{ij}$ (a); Wheels $W_i$, $W_l$, $W_j$, $W_k$ ((b), c.c.w) 
with shadowed areas of the triangles participating in 
the definition of \emph{large atomic stresses} for $W_l$ and $W_j$.}
\label{fig:wheelDecomposition}
\end{figure}
The decomposition of an equilibrium stress into a linear combination of ``local'' equilibrium stresses will be the key 
for the next embedding algorithm. The Wheel-Decomposition Theorem presented in this subsection provides the underlying theory.

Before proceeding, let us review how the canonical equilibrium stress associated with an orthogonal projection of a 3d polytope 
in the $\{z=0\}$ plane can be described. The
assignment of heights to the interior vertices of a 2d embedding resulting in a polyhedral surface is called a (polyhedral) \emph{lifting}. By the Maxwell--Cremona correspondence 
the equilibrium stresses of a 2d embedding of a planar 3-connected graph and its liftings
(modulo arbitrary choice of lifting for any one face of the graph) 
are in 1-1 correspondence. 
Moreover, the bijection between liftings and stresses can be defined as follows. Let $\mathbf p$ be a 2d drawing of a triangulation 
and let $\mathbf u$ be the 3d embedding induced by some lifting. We map this lifting to the equilibrium stress $\omega$
by assigning to every edge $(v_iv_j)$ separating the faces $(v_iv_jv_k)$ (on the left) and $(v_iv_jv_l)$ (on the right)
\begin{equation} \label{eq:ProjEqui}
\omega_{ij}:=\frac{[u_iu_ju_ku_l]}{[p_ip_jp_k][p_lp_jp_i]}.
\end{equation}
We refer to this stress as the {\em canonical equilibrium stress} of a projection.
This mapping gives the desired bijection. The expression of Eq.~\eqref{eq:ProjEqui} 
is a slight reformulation of the form presented in Hopcroft~and~Kahn~\cite[Equation~11]{HK92}.
We note that the canonical equilibrium stress is defined only when the denominators
participating in Eq.\eqref{eq:ProjEqui} are nonzero, 
that is when none of the faces of the projection $\mathbf p$ is degenerate.

We continue by studying the spaces of equilibrium stresses for triangulations. 
A graph formed by a vertex $v_c$, called \emph{center}, 
connected to every vertex of a cycle $v_1,\ldots,v_n$, called {\em base},
is called a \emph{wheel} (no other edge is present, see Fig.~\ref{fig:wheelDecomposition}); we denote it as 
$W(v_c;v_1\ldots v_n)$. 
A wheel that is a subgraph of a triangulation $G$ with $v_i\in V(G)$ as a center is denoted by $W_i$.
Every triangulation can be  ``covered'' with a set of wheels $(W_i)_{v_i\in V(G)}$, 
so that every edge is covered four times (Fig.~\ref{fig:wheelDecomposition}).

\begin{lemma}\label{lm:pyrstress}
Let $\mathbf p = (p_c,p_1,\ldots p_n)$ be an embedding of a wheel $W(v_c;v_1\ldots v_n)$ in $\R^2$,
such that for every $1\le i\le n$ the points $p_i$, $p_{i+1}$ and $p_c$ are noncollinear
(we use the cyclic notation for the vertices of the base of the wheel).
Then the following expression defines an equilibrium stress:
$$
\omega_{ij}=
\begin{cases}
   -1/[p_ip_{i+1}p_c],  \quad &j=i+1, 1\le i\le n,\\
   [p_{i-1}p_{i}p_{i+1}]/([p_{i-1}p_ip_c][p_ip_{i+1}p_c]), \quad &j=c, 1\le i\le n.
\end{cases}
$$
The equilibrium stress for the embedding $\mathbf p$ is unique up to a renormalization.
\end{lemma}
\begin{proof}
 Let $W^l$ be the lifting of $W$ such that $z_c=1$ and $z_i=0$ for $1\le i\le n$.
 Then $W$ is the orthogonal projection of $W^l$ 
 and we can compute the canonical equilibrium stress given by Eq.~\eqref{eq:ProjEqui} on it. 
 However, the canonical equilibrium stress is defined only for simplicial polytopes.
 To make it applicable to the wheel, we arbitrarily triangulate the possibly nontriangular ``base face" $p_1,\ldots p_n$.
 Since the newly introduced edges are flat in the lifting $W^l$, 
 the canonical equilibrium stress is zero on these edges. 
 Thus only the original edges of $W$ carry the stress.
 To finish the proof we note that the stress from the statement of the lemma exactly coincides
 with the computed canonical equilibrium stress, thus it is an equilibrium stress.
 The space of equilibrium stresses on a wheel is 1-dimensional,
 since the space of polyhedral liftings of the wheel is 1-dimensional.
\end{proof}

\begin{definition}
\begin{enumerate}
\item In the setup of Lemma~\ref{lm:pyrstress}, we call the equilibrium stress $\omega$
the \emph{small atomic stress} of the wheel $W$ and denote it as $\omega^a(W)$.
\item We call the stress  $\omega^A(W)$ that is  obtained by the multiplication of $\omega^a(W)$ by the factor
 $\prod_{1\le j\le n} [p_jp_{j+1}p_c]$, the \emph{large atomic stress} of $W$. 
\end{enumerate}
\end{definition}
We note that the large atomic stresses are products of $\deg(v_c)-1$
triangle areas multiplied by $2$, and thus, 
all stresses $\omega_{ij}^A(W)$ are integers if $W$ is realized with integer coordinates.

\begin{theorem}[Wheel-Decomposition Theorem for Equilibrium Stresses]
\label{thm:wheelDecStress}
 Let $G$ be a triangulation and $\mathbf p=(p_i)_{1\le i\le n}$ be an embedding of $G$ to $\R^2$
 such that every face of $G$ is realized as a nondegenerate triangle.
 Then every equilibrium stress $\omega$ on $\mathbf p$
 can be expressed as a linear combination of the small atomic 
 stresses on the wheels $(W_i)_{1\le i\le n}$
 $$
 \omega=\sum_{1\le i\le n}\alpha_{i}\omega^a(W_i).
 $$
The decomposition is not unique:
 valid coefficients $\alpha_i$ are given by the heights (i.e., $z$-coordinates) of the corresponding vertices $p_i$
in any of the Maxwell--Cremona liftings of $\mathbf p$ induced by $\omega$. 
\end{theorem}
\begin{proof}

Maxwell--Cremona lifting procedure defines a linear isomorphism between 
the linear spaces of equilibrium stresses and of polyhedral liftings of a planar embedding of a graph.

As shown in Lemma~\ref{lm:pyrstress}, the atomic stress $\omega^a(W_i)$ corresponds to a lifting that lifts the vertex
$p_i$ to the height 1 and leaves all the other vertices untouched. Thus the stress that lifts the $i$-th vertex to $z_i$ is
exactly the linear combination $\sum_{1\le i\le n}z_i \omega^a(W_i)$.

\end{proof}

\subsection{An Efficient Reverse of the Maxwell--Cremona Lifting}
\label{subsec:WDreverseMC:reverseMC}

%
 The direct way to reverse the Maxwell--Cremona lifting procedure would be to project
 the 3d embedding of a graph to a plane and to calculate stresses using Eq.~\eqref{eq:ProjEqui}. If the 3d embedding has polynomial size integer coordinates,
 the calculated stresses are, generally speaking, rational, and 
 when scaled to integers may increase by an exponential factor.
 The following theorem provides a more careful method by allowing a small perturbation of the canonical equilibrium stress as 
 given by Eq.~\eqref{eq:ProjEqui}.
 
 %

\begin{theorem}[Reverse of the Maxwell--Cremona Lifting]
 \label{thm:reverseMC}
Let $\mathbf u = (u_i)_{1\le i\le n}$ be an embedding of a triangulation $G$ into $\Z^3$
such that none of the planes supporting the faces of $G$ is orthogonal to the plane $\{z=0\}$
and no pair of adjacent faces is coplanar.
Let $\mathbf p = (p_i)_{1\le i\le n}$ be the orthogonal projection of $\mathbf u$ to the plane $\{z=0\}$.
Then one can construct an integer equilibrium stress $\omega$ 
on $\mathbf p$ such that
$$
|\omega_{ij}|<8\cdot (2\max_{i\le n}|u_i|)^{2\maxdeg+5}
$$
and $\sign(\omega_{ij})=\sign(\widetilde\omega_{ij})$ for the canonical equilibrium stress $\widetilde\omega$ on $\mathbf p$ as defined by Eq.~\eqref{eq:ProjEqui}. 
\end{theorem}

\begin{proof}
Let $\gridsize = \max_{i,j}|u_i-u_j|$ be the size of the grid containing the initial embedding.
We start with the canonical equilibrium stress $\widetilde\omega$ as specified by Eq.~\eqref{eq:ProjEqui} 
 for the embedding $\mathbf p$. Since all the coordinates are integers, 
 and the embedding $\mathbf u$ has no flat edges, all stresses are bounded by
$$\frac{1}{\gridsize^4}\le \frac{1}{|[p_ip_jp_k]||[p_lp_jp_i]|}\le|\widetilde\omega_{ij}|\le|[u_iu_ju_ku_l]|\le \gridsize^3.$$

We are left with making these stresses integral 
while preserving a polynomial bound.
The faces of $\mathbf u$ are nonvertical, thus the faces of $\mathbf p$ are nondegenerate and
we may apply the Wheel-Decomposition Theorem to the stress $\widetilde\omega$:
we rewrite it as a linear combination of the \emph{large} atomic stresses 
of the wheels,
 $$
 \widetilde\omega = \sum_{1\le k\le n}\alpha_k \omega^A(W_k).
 $$
We remark that since we use \emph{large} atomic stresses, coefficients $\alpha_k$ are not the $z$-coordinates of $u_k$,
but these $z$-coordinates divided by the multiplicative factor from the definition of the large atomic stress.
 Since all the points $p_i$ have integer coordinates, the large atomic stresses are integers as well.
 Moreover, each of them, as a product of $\deg(v_k)-1$ triangle areas, is bounded by
 $|\omega^A_{ij}(W_k)|\le \gridsize^{2(\maxdeg-1)}$. 
 
 To make the $\widetilde\omega_{ij}$s integral  we round the coefficients $\alpha_k$ down.
To guarantee that the rounding does not alter the signs of the stress,
 we scale the atomic stresses (before rounding) with the factor
 $$C=4\max_{i,j,k}|\omega_{ij}^A(W_k)|/\min_{i,j}|\widetilde\omega_{ij}|$$
 and define as the new stress:
  $$
  \omega := \sum_{1\le k\le n} \lfloor C\alpha_k \rfloor \omega^A(W_k).
  $$
 Clearly, 
  \begin{align*}
  |\omega_{ij}-C\widetilde\omega_{ij}|
  &=\left|\sum_{1\le k\le n}(\lfloor C\alpha_k\rfloor\ -C\alpha_k)\omega_{ij}^A(W_k)\right|\\
  &<\sum_{1\le k\le n}|\omega_{ij}^A(W_k)|
  \le4\max_{i,j,k}|\omega_{ij}^A(W_k)|
  =C\min_{i,j}|\widetilde\omega_{ij}|\le C|\widetilde\omega_{ij}|,
  \end{align*}
where the third inequality holds since exactly four wheels participate in the
wheel decomposition of every single edge.
Thus, $\sign(\omega_{ij})=\sign(C\widetilde\omega_{ij})=\sign(\widetilde\omega_{ij})$. From the last equation it also follows that none of the stresses 
$\omega_{ij}$ are zero.

 Therefore,  the constructed equilibrium stress $\omega$ is integral and of the same sign structure as the canonical equilibrium stress.
  We conclude the proof with an upper bound on its size.
 Since $C<4\gridsize^{2(\maxdeg-1)} \gridsize^{4}$, 
 \begin{align*}
|\omega_{ij}|
&\le \left|\sum_{1\le k\le n}(C\alpha_k\pm1)\omega_{ij}^A(W_k)\right|
  \le C|\widetilde\omega_{ij}| \!+\!\sum_{1\le k\le n} |\omega_{ij}^A(W_k)|\\ 
&\le C\max|\widetilde\omega_{ij}| \!+\! 4\max|\omega_{ij}^A(W_k)|
 \le 4\gridsize^{2\maxdeg+2}\cdot \gridsize^3 + 4\gridsize^{2\maxdeg-2}
 \le 8\gridsize^{2\maxdeg+5}.
  \end{align*}
\end{proof}

\subsection{The duality transform}
\label{subsec:WDreverseMC:duality}
As an instant by-product of the techniques developed in this section
we can now prove the duality transform for simplicial polytopes with
special geometry and a degree-3 vertex.
%
\begin{theorem}
\label{thm:simplicialDualitySpecial}
Let $G_{\uparrow}$ be a triangulation and let $\mathbf u = (u_i)_{2\le i\le n}$ be its realization
as a convex polytope with integer coordinates, such that
its orthogonal projection into the plane $\{z=0\}$ 
is a planar embedding $(p_i)_{2\le i\le n}$ of $G_{\uparrow}$ with 
boundary face $(v_2v_3v_4)$.
Then the exists a realization
$(\phi_{f})_{f\in F(G)}$ of a graph dual to $G=\Stack(G_{\uparrow};v_1;v_2v_3v_4)$ 
with integer coordinates bounded by
$$
|\phi_f|=O(n\max|u_i|^{2\maxdeg+7}).
$$
%
\end{theorem}
\begin{proof}
By combining Theorem~\ref{thm:frameworkL} with Theorem~\ref{thm:reverseMC} we obtain the statement of the theorem.
\end{proof}

We remark that the algorithms following the lifting approach, e.g.~\cite{RG96},
generate embeddings that fulfill the conditions of the above theorem. 

The construction presented in this section can
be modified for general simplicial polytopes
without restrictions on the geometry
and without a degree-3 vertex.
However, such a modification requires a substantial amount
of technicalities and leads to much worse estimates
than the 3-dimensional techniques developed in the next section.

\section{A Duality Transform for General Simplicial polytopes}
\label{sec:simplicialDuality}

Unlike the transforms for stacked polytopes (Theorem~\ref{thm:truncated}) and for a special 
case of simplicial polytopes with
a degree-3 vertex (Theorem~\ref{thm:simplicialDualitySpecial}),
the algorithms of this section are intrinsically 3-dimensional and
use neither planar (flat) embeddings in intermediate steps, nor 2-dimensional equilibrium stresses.

\subsection{Canonical CDV matrices}
\label{subsec:simplicialDuality:canonicalCDV}

We begin with an accurate description of the {\em canonical CDV matrix}
for a spacial embedding of a graph, first shortly mentioned in Sect.~\ref{sec:lovasz}.
The following construction is due to Lov{\'a}sz~\cite{L01},
we cite the proof due to its simplicity.

\begin{lemma}[\cite{L01}, Sect.5]
\label{lem:6}
Let $\P=(u_i)_{1\le i\le n}$ be an embedding of a graph $G$
into $\R^3$ as a polygonal surface (straight-line embedding
with each face realized as a flat polygon) such that none of the planes
supporting the faces of $\P$ passes through the origin.
Let $(\phi_f)_{f\in F(G)}$ be the vectors 
normal to the faces of $\P$ normalized so that 
$\langle\phi_f, x\rangle=1$ for every point $x$ of the face $f$.

Then there exists a unique CDV matrix for $G$
such that Eq.~\eqref{eq:phi} holds
for each pair of dual edges $(u_i,u_j)$ and $(\phi_f,\phi_g)$.
\end{lemma}
\begin{definition}
We call a matrix $M$ defined above the {\em canonical CDV matrix} of an embedding.
\end{definition}

\paragraph{\it Proof of Lemma~\ref{lem:6}}
First, we check that the left and right hand sides of Eq.~\eqref{eq:phi} are parallel vectors and $u_i\times u_j \ne 0$ and thus the equation correctly defines $M_{ij}$ for $(i,j)\in E(G)$. Indeed, for $k\in \{i,j\}$
$$
\langle\phi_f-\phi_g, u_{k}\rangle=\langle\phi_f, u_{k}\rangle-\langle\phi_g, u_{k}\rangle=1-1=0
$$
by the chosen normalization of normals and
$$
\langle u_i\times u_j, u_{k}\rangle=0
$$
by the definition of dot product.
The vectors $u_i$ and $u_j$ span a plane since, 
if they were parallel, planes supporting both faces $f$ and $g$
would pass through the origin.
So, both sides of Eq.~\eqref{eq:phi} are orthogonal to the plane
spanned by the vectors $u_i$ and $u_j$ and thus are parallel.
Since the vectors $u_i$ and $u_j$ are not parallel, $u_i\times u_j\ne 0$.
Thus, Eq.~\eqref{eq:phi} uniquely defines $M_{ij}$
for $(i,j)\in E(G)$.

Second, we set the diagonal $(M_{ii})_{1\le i\le n}$ so that the CDV equilibrium condition holds. 
We can always achieve this due to the following observation:
First note that
\begin{align*}
\left(\sum_{v_j\in N(G,v_i)}M_{ij}u_j \right)\times u_i 
=& \!\!\! \sum_{v_j\in N(G,v_i)} \!\!\!\! M_{ij}(u_j\times u_i)= \sum_{1\le k\le \deg(v_i)}\phi_{f_k}-\phi_{f_{k+1}} = 0,
\end{align*}
where $(\phi_1, \phi_2, \ldots, \phi_{\deg(v_i)})$ is 
the cyclic sequence of faces incident to $v_i$. 
So,
$$
\sum_{v_j\in N(G,v_i)}M_{ij}u_j \parallel u_i 
$$
and since $u_i\ne 0$ there exists a unique $M_{ii}$ such that
$$
\sum_{v_j\in N(G,v_i)}M_{ij}u_j = -M_{ii}u_i.
$$
All the off-diagonal nonedge elements of $M$ are filled with zeroes.
\endproof

\subsection{CDV matrices and equilibrium stresses} 
\label{subsec:simplicialDuality:equiCDVproj}

We show now that the linear spaces of CDV matrices and of equilibrium stresses are indeed isomorphic
by presenting a projective transformation between these two spaces.
This correspondence is also of independent interest
as it highlights the projective nature of the notion of equilibrium stress.
%
%
We start by observing that the CDV matrices can be arbitrarily rescaled together with
the embedding. This result was already noted by Lov{\'a}sz~\cite{L01}; we include the proof for completeness.
The graphs in this subsection are not necessarily be planar nor 3-connected.
\begin{lemma}
\label{lm:CDVscalable}
Let $\mathbf u = (u_i)_{1\le i\le n}$ and $\mathbf r = (r_i)_{1\le i\le n}$ be two embeddings of a graph $G$
into $\R^3$ such that 
$$r_i = \lambda_i u_i, \quad \lambda_i\in \R\setminus\{0\}, \,1\le i \le n.$$

Then the map $\pr_{u\to r}:\R^{n\times n}\to \R^{n\times n}$ defined as

$$
(\pr_{u \to r}M)_{ij} = \frac1{\lambda_i\lambda_j} M_{ij}, \qquad 1\le i,j\le n
$$
is a linear isomorphism between the linear spaces of CDV matrices of $\mathbf u$ and $\mathbf r$ with 
$
\pr_{u\to r}\pr_{r\to u} = id.
$
\end{lemma}

\begin{proof}

We first show that for every CDV matrix $M$ of $\mathbf u$ the image 
$\pr_{u\to r}(M)$ is a CDV matrix for $\mathbf r$.
Indeed,
$$
\sum_{j=1}^{n} \frac1{\lambda_i\lambda_j}M_{ij} r_j =
\frac1{\lambda_i}\sum_{j=1}^{n}M_{ij}\frac1{\lambda_j}r_j =
\frac1{\lambda_i}\sum_{j=1}^{n}M_{ij}u_j = 0 \quad 1\le i\le n.
$$

Next, trivially, $\pr_{u\to r}\pr_{r\to u} = id$:

$$
(\pr_{r\to u}\pr_{u\to r} M)_{ij} = 
\lambda_i\lambda_j (\pr_{u\to r} M)_{ij} = 
\lambda_i\lambda_j\frac1{\lambda_i\lambda_j}  M_{ij} = 
M_{ij}.
$$

Finally, in the matrix form
$$
\pr_{u\to r}M = 
\begin{pmatrix}
1/\lambda_1,& \ldots, & 1/\lambda_n
\end{pmatrix} M 
\begin{pmatrix}
1/\lambda_1 \\ 
\cdots \\
1/\lambda_n
\end{pmatrix}
$$
and thus $\pr_{u\to r}$ is linear.
\end{proof}

The next lemma establishes an isomorphism between the linear spaces of CDV matrices and equilibrium stresses:
\begin{lemma}
\label{lm:equiCDVprojcorr}
Let $\mathbf u = (u_i)_{1\le i\le n}$ be an embedding of a graph $G$ in $\R^3$ such that none of its vertices is at the origin.
Let $\alpha$ be any plane that does not contain the origin and is not parallel to any of the vectors $u_i$.
Let $\mathbf p = (p_i)_{1\le i\le n}$ be the central projection of $\mathbf u$ to $\alpha$:
$$
p_i := \lambda_i u_i, \quad \lambda_i\in \R,\,p_i\in \alpha.
$$

Then
the map $\pr_{\alpha}: \R^{n\times n}\to \R^{E(G)}$ defined as
$$
(\pr_{\alpha} M)_{ij} = \frac1{\lambda_i}\frac1{\lambda_j} M_{ij},\quad (ij)\in E(G)
$$
is a linear isomorphism between the linear space of CDV matrices of $\mathbf u$ and the linear space of equilibrium stresses of $\mathbf p$.

In particular,
\begin{itemize}
\item for any CDV matrix $M$ of $\mathbf u$ the assignment
$$
\omega_{ij} := \frac1{\lambda_i}\frac1{\lambda_j}M_{ij}, \quad (ij)\in E(G)
$$
is an equilibrium stress for $\mathbf p$, and
\item for any equilibrium stress $\omega$ of $\mathbf p$ the assignment
$$
M_{ij} := \lambda_i\lambda_j\omega_{ij}, \quad (ij)\in E(G)
$$
can be extended in a unique way to a CDV matrix $M$ for $\mathbf u$.
\end{itemize}

\end{lemma}

\begin{proof}
Due to Lemma~\ref{lm:CDVscalable}, the map 
$\pr_{u\to p}:\R^{n\times n}\to \R^{n\times n}$
defined as $\pr_{u\to p}(M)_{ij} = \frac1{\lambda_i\lambda_j}M_{ij}$
is a linear isomorphism between the spaces of CDV matrices of $\mathbf u$ and $\mathbf p$.

Due to Lemma~\ref{lm:equiCDV} (parts~\ref{lm:equiCDV:2} and~\ref{lm:equiCDV:3}), 
the map $f: \R^{n\times n}\to \R^{E(G)}$
that is identical on the edge-entries of a CDV matrix and forgets all the nonedge entries, $f(M)_{ij} = M_{ij}$ for $(ij)\in E(G)$, is a linear isomorphism between the space of CDV matrices of $\mathbf p$ and the space of equilibrium stresses on $\mathbf p$.

To finish the proof we remark that $\pr_{\alpha} = f\cdot \pr_{u\to p}$.
\end{proof}

\subsection{Wheel-decomposition for CDV matrices}
\label{subsec:simplicialDuality:WD}
We use the correspondence between equilibrium stresses
and CDV matrices to define the concept of {\em atomic CDV matrices}
and to formulate and prove the {\em Wheel-Decomposition Theorem} for CDV matrices, 
which is the analogue of Theorem~\ref{thm:wheelDecStress} for CDV matrices.

\begin{lemma}
\label{lm:CDVwheel}
Let $W = W(v_c; v_1, \ldots v_n)$ be a wheel with the center $v_c$.
Let $\mathbf u = (u_c,u_1,\ldots u_n)$ be an embedding of $W$ to $\R^3$
such that for every $1\le i\le n$ the points $u_i$, $u_{i+1}$ and $u_c$ are noncoplanar with the origin
(as usual, we use the cyclic notation for the vertices of the base of the wheel).
Then 
\begin{enumerate}
\item\label{lm:CDVwheel:1} The assignment
$$
M^a_{ij}:=
\begin{cases}
   -\frac{1}{\det(u_iu_{i+1}u_c)},  \quad &j=i+1, 1\le i\le n,\\
   \frac{\det(u_{i-1}u_{i}u_{i+1})}
   {\det(u_{i-1}u_{i}u_c)\det(u_iu_{i+1}u_c)}, \quad &j=c, 1\le i\le n
\end{cases}
$$
can be extended in a unique way to a CDV matrix $M^a$ for $\mathbf u$;
\item\label{lm:CDVwheel:2} A CDV matrix for $\mathbf u$ is unique up to scaling;
\item\label{lm:CDVwheel:3} Let $\alpha$ be any affine plane at the distance 1
from the origin not parallel to any of $u_i$.
Let $\mathbf p$ be the central projection of $\mathbf u$ to this plane, $p_i=\lambda_i u_i\in \alpha,\,\lambda_i\in\R$. 
Then, in the notation of Lemma~\ref{lm:equiCDVprojcorr},
\begin{align*}
M^{a} &=\lambda_c \pr_{\alpha}^{-1}(\omega^a),\\
M^{A} &= \frac1{\lambda_c^n (\prod_{1\le i\le n} \lambda_i)^2} \pr_{\alpha}^{-1}(\omega^A),
\end{align*}
where $\omega^a$ and $\omega^A$ are the {\em small} and {\em large} 
atomic stresses for $\mathbf p$ correspondingly.
To compute the atomic stresses for $\mathbf p$ we suppose 
that the plane $\alpha$ is oriented in the same way 
as the plane $\{z=1\}$.
\end{enumerate}
\end{lemma}

\begin{definition}
\label{def:atomicCDV}
We call the CDV matrix $M^a$ defined in Lemma~\ref{lm:CDVwheel} the {\em atomic CDV matrix} for
the embedding $\mathbf u$.

We call the rescaling of $M^a$ with the factor
$\prod_{1\le i\le n} \det(u_iu_{i+1}u_c)$ the 
{\em large atomic CDV matrix} and denote it with
$$M^{A} = \left( \prod_{1\le i\le n} \det(u_iu_{i+1}u_c) \right) \cdot M^a .$$
\end{definition}

Note that for an integer embedding $\mathbf u$ the large atomic CDV matrix $M^{A}$ is an integer matrix.

\begin{proof}[Proof of Lemma~\ref{lm:CDVwheel}]
Pick any affine hyperplane $\alpha$ at the distance 1 from the origin
such that none of $u_i$ is parallel to $\alpha$.
Let $\mathbf p$ be the central projection of $\mathbf u$ to this plane, $p_i=\lambda_i u_i \in \alpha$.
Due to Lemma~\ref{lm:equiCDVprojcorr} the CDV matrices of $\mathbf u$ are in 1-to-1 correspondence with 
the equilibrium stresses of $\mathbf p$, which proves the uniqueness (part~\ref{lm:CDVwheel:2} of the lemma).

To finish the proof we check the expressions in part~\ref{lm:CDVwheel:3}.
Due to Lemma~\ref{lm:equiCDVprojcorr},  $\pr_{\alpha}^{-1}(\omega^a)$
is a CDV matrix, which also proves part~\ref{lm:CDVwheel:1} of the lemma.
We check only the equation for small atomic stresses. 
The case of large atomic stresses can be checked similarly. 
Since none of the triples of points $u_i$, $u_{i+1}$, $u_c$ for $1\le i\le n$
is coplanar with the origin, none of the triples of points $p_i$, $p_{i+1}$, $p_c$ is collinear.
Thus we can construct the small atomic stress for $\mathbf p$: 
$$
\omega^a_{ij}:=
\begin{cases}
   -\frac{1}{[p_ip_{i+1}p_c]},  \quad &j=i+1, 1\le i\le n,\\
   \frac{[p_{i-1}p_{i}p_{i+1}]}
   {[p_{i-1}p_{i}p_c][p_ip_{i+1}p_c]}, \quad &j=c, 1\le i\le n.
\end{cases}
$$
In the formulas above we view $p_i$ as points in the plane $\alpha$.
Since the area is translationary invariant, the choice of the origin within $\alpha$ does not play a role and the areas $[p_kp_lp_m]$ are well defined.
By Lemma~\ref{lm:equiCDVprojcorr}, the image of $\omega^a$ under the reverse projection 
$M:=\pr_{\alpha}^{-1}(\omega^a)$
of $\mathbf p$ to $u_i=\frac1{\lambda_i}p_i$ is a CDV matrix for $\mathbf u$ with
$$
M_{ij}:=\lambda_i\lambda_j\omega^a_{ij}.
$$
We additionally rescale $M$ with the factor $\lambda_c$ to 
$$
M^a:=\lambda_c M
$$
and remark that
$$
\frac{[p_ip_kp_l]}{\lambda_i\lambda_k\lambda_l}=\det(u_iu_ku_l).
$$
A straightforward computation finishes the proof.
\end{proof}

The Wheel-Decomposition Theorem for CDV matrices 
is now in context of Lemma~\ref{lm:equiCDVprojcorr} 
a straightforward consequence of 
the Wheel-Decomposition Theorem for equilibrium stresses 
as given in Theorem~\ref{thm:wheelDecStress}:
\begin{theorem}[Wheel-Decomposition Theorem for CDV Matrices]
\label{thm:wheelDecCDV}
Let $G$ be a triangulation with $n>3$ vertices.
Let $\P=(u_i)_{1\le i\le n}$ be an embedding of $G$ in $\R^3$
such that none of the planes supporting the faces of $\P$ passes through the origin.
Let $M$ be a CDV matrix for $\P$.
Then there exists a set of real coefficients $(\alpha_i)_{1\le i\le n}$ such that
$$M = \sum_{1\le i\le n} \alpha_i M^a(W_i),$$
where $M^a(W_i)$ is the atomic CDV matrix for the wheel $W_i\subset G$ centered at the vertex $v_i$.
The set of coefficients is unique up to the ``parallel translation" transformation 
$$
\alpha_i \to \alpha_i + \langle v, u_i \rangle \qquad \forall i
$$
for any vector $v$ in $\R^3$.
\end{theorem}
\begin{proof}
Let $\alpha$ be any affine plane that does not contain the origin 
and is not parallel to any of the vectors $u_i$.
Let $\mathbf p = (p_i)_{1\le i\le n}$ be the central projection of $\P$ to this plane and 
$\omega=\pr_{\alpha}(M)$ be the central projection of $M$ (see Lemma~\ref{lm:equiCDVprojcorr}):
\begin{align*}
p_i &=\lambda_i u_i:\quad p_i \in \alpha,\quad &&1\le i\le n,\\
\omega_{ij} &= \frac1{\lambda_i\lambda_j}M_{ij},\quad &&(i,j)\in E(G).
\end{align*}

Due to Lemma~\ref{lm:equiCDVprojcorr}, $\omega$ is an equilibrium stress for $\mathbf p$.
Since the faces of $\P$ are not coplanar with the origin, the faces of $\mathbf p$ are nondegenerate triangles.
Thus we can apply the Wheel-Decomposition Theorem for equilibrium stresses to get
\begin{equation}
\label{eq:localwheel}
\omega = \sum_{1\le k\le n} \alpha_k \omega^a(W_k),
\end{equation}
where $\omega^a(W_k)$ is the small atomic stress for the wheel centered at the vertex $v_k$.
We use the second part of Lemma~\ref{lm:equiCDVprojcorr} and project Eq.~\eqref{eq:localwheel} back to $\P$:
$$
\pr_{\alpha}^{-1}\omega = 
\sum_{1\le k\le n} \alpha_k  \pr_{\alpha}^{-1}\omega^a(W_k).
$$
By construction, the left-hand side equals $M$.
By the third part of Lemma~\ref{lm:CDVwheel}, the summands on the right-hand side equal the rescaled atomic CDV matrices of the wheels,
$$
\pr_{\alpha}^{-1}\omega^a(W_k) = 
\frac1{\lambda_k} 
M^a (W_k),\quad 1\le k\le n.
$$
Thus,
$$
M=\sum_{1\le k\le n} \pr_{\alpha}^{-1}(\alpha_k) M^a(W_k),
$$
where $\pr_{\alpha}^{-1}(\alpha_k) = \frac{\alpha_k}{\lambda_k}$.

To finish the proof we analyze the uniqueness of the coefficients.
The coefficients in the planar decomposition
\eqref{eq:localwheel} are the heights of the corresponding vertices in any of the Maxwell--Cremona liftings of $\mathbf p$  with help of $\omega$.  
They are unique up to the transformation 
$\alpha_k \to \alpha'_k = \alpha_k + \langle v, p_k\rangle $, where $v$ is any vector in $\R^3$.
The corresponding transformation for the coefficients of the 
wheel-decomposition of CDV matrices is:
\begin{align*}
\pr_{\alpha}^{-1}(\alpha_k) 
\to \pr_{\alpha}^{-1}(\alpha'_k) = \frac{\alpha'_k}{\lambda_k} 
=& \frac{\alpha_k+\langle v, p_k \rangle}{\lambda_k}\\
=& \pr_{\alpha}^{-1}(\alpha_k)+\langle v, \frac{p_k }{\lambda_k}\rangle
=\pr_{\alpha}^{-1}(\alpha_k)+\langle v, u_k \rangle.
\end{align*}
\end{proof}

\subsection{The duality transform}
\label{subsec:simplicialDuality:simplicialDuality}
We can now wrap up to conclude with the main theorem of this section. We start by introducing a theorem that presents the crucial
step in the dual transform for general simplicial polytopes.
\begin{theorem}
\label{thm:CDVsimplicial}
Let $\P=(u_i)_{1\le i\le n}$ be a convex simplicial polytope in $\R^3$ 
with integer coordinates, containing the origin in its interior.
Then there exists an integral positive CDV matrix $M$ for $\P$ such that its entries are bounded by
$$
|M_{ij}| = O(\max_{1\le i\le n}|u_i| ^ {3\maxdeg + 7}),
$$
where $\maxdeg$ is the maximal vertex degree in $G$.
\end{theorem}
\begin{proof}
The proof goes through 3 steps:
\begin{enumerate}
\item[(I)] We construct as $M^c$ the canonical CDV matrix for $\P$ and bound its entries.
\item[(II)] We use the Wheel-Decomposition Theorem to decompose $M^c$
as a linear combination of the {\em large} atomic CDV matrices $M^c = \sum_{1\le k\le n}\alpha_k M^A(W_k)$ and we bound the entries of $M^A(W_k)$,
\item[(III)] We define the final CDV matrix  $M := \sum_{1\le k\le n}\lfloor C\alpha_k\rfloor M^A(W_k)$ with $C=4\frac{\max(M^A_{ij}(W_k))}{\min(M^c_{ij})}$ and check its correctness.
\end{enumerate}
Throughout the proof we set $\gridsize := \max_{ij}|u_i-u_j|$.\\
Step (I): Let $M^c$ be the canonical CDV matrix for $\P$:
$$
\phi_f-\phi_g = M^c_{ij}(u_i\times u_j)
$$
for every pair of dual edges $(u_iu_j)$ and $(\phi_f\phi_g)$.
In the next steps we need to bound $|M^c_{ij}|$. 
Obviously,
$$ \frac{\max|\phi_f-\phi_g|}{\min|u_i\times u_j|}  \ge |M^c_{ij}| \ge \frac{\min|\phi_f-\phi_g|}{\max|u_i\times u_j|}. 
$$
To bound $|M^c_{ij}|$ further we need a lower bound for $|\phi_f-\phi_g|$ (estimate a) and an upper bound for  $|\phi_f|$ (estimate b).

(estimate a) We remark that
$\phi_f$ is the solution of the linear system
$$
\langle\phi_f, u^f_1\rangle = 1, 
\langle\phi_f, u^f_2\rangle = 1,
\langle\phi_f, u^f_3\rangle = 1,$$
where the vertices $u^f_1u^f_2u^f_3$ belong to the face $f$.
Thus, since $u_i$ are integral,
$$\phi_f = \frac1{\det(u^f_1u^f_2u^f_3)} Z_3$$
for $Z_3$ being some vector in $\mathbb{Z}^3$.
Similarly, 
$$\phi_g = \frac1{\det(u^g_1u^g_2u^g_3)} Z'_3.$$
Again, $Z'_3$ is some vector in $\mathbb{Z}^3$.
Thus, going to the common denominator,
$$|\phi_f-\phi_g| = \frac1{\det(u^f_1u^f_2u^f_3)
\det(u^g_1u^g_2u^g_3)} Z''_3$$
for some third $Z''_3$ in $\mathbb{Z}^3$.
So, since $\phi_f\ne \phi_g$,
$$
|\phi_f-\phi_g|\ge \frac1{\det(u^f_1u^f_2u^f_3)
\det(u^g_1u^g_2u^g_3)}
\ge \frac1{\gridsize^6}.
$$

(estimate b) To bound $|\phi_f|$ from above we 
consider the pyramid over the face $f=(u^f_1u^f_2u^f_3)$ with the tip at $\mathbf{0}=(0,0,0)^T$.
Its (nonoriented) volume $\Vol(\mathbf{0},u^f_1,u^f_2,u^f_3)$ can be computed as
$$
\frac13 h_f \Area(u^f_1u^f_2u^f_3) \
= \Vol(\mathbf{0},u^f_1,u^f_2,u^f_3),
$$
where $h_f$ is the height to the base $f$. 
Since all $u_i$s are integral, $\Vol(\mathbf{0},u^f_1,u^f_2,u^f_3)\ge \frac{1}{6}$.
Trivially, $\Area(u^f_1u^f_2u^f_3)\le L^2$.
Since $\phi_f$ is the polar vector for the plane supporting the face $f$, 
and $|h_f|$ is the distance from the origin to this plane,
$|h_f| = \frac1{|\phi_f|}$.
Thus,
$$
|\phi_f| \le 
\frac{\Area(u^f_1u^f_2u^f_3)}{3\Vol(\mathbf{0},u^f_1,u^f_2,u^f_3)} \le 2L^2.
$$

Using the two bounds of (estimate a) and (estimate b) gives: 
\begin{align}
|M^c_{ij}| 
	&\ge \frac{\min|\phi_f-\phi_g|}{\max|u_i\times u_j|} 
	\ge \frac{1/\gridsize^6}{\gridsize^2} 
	= \frac1{\gridsize^8}, \label{eq:boundC2}\\
|M^c_{ij}| 
	&\le \frac{\max|\phi_f-\phi_g|}{\min|u_i\times u_j|} 
	\le 2\max|\phi_f|
	\le 4\gridsize^2.\nonumber
\end{align}

%

Step (II): We use the Wheel-Decomposition Theorem to decompose $M^c$  into 
a linear combination of the {\em large} atomic CDV matrices of wheels:
\begin{equation}
\label{eq:localwheel2}
M^c = \sum_{1\le k\le n}\alpha_k M^A(W_k).
\end{equation}
As a next step we bound the entries $|M^A_{ij}|$ from above. By definition,
\begin{align}
\label{eq:boundC1}
|M^A_{ij}(W_k)| 
	&\le (\max_{o,p,q}|\det(u_ou_pu_q)|)^{\deg(v_k)-1} 
	\le \gridsize^{3(\deg(v_k)-1)}.
\end{align}

Step (III):
The large atomic CDV matrices in the decomposition~\eqref{eq:localwheel2} are integers. 
To make the whole sum integral we round the coefficients. 
Though, direct rounding may influence the sign of the resulting stresses. 
To overcome this effect we scale the coefficients before rounding.
We pick the scaling factor 
$$
C := 4\left\lceil \frac{\max_{i,j,k} |M^A_{ij}(W_k)|}{\min_{i,j} |M^c_{ij}|}\right\rceil
$$
and construct 
$$
M := \sum_{1\le k\le n} \lfloor C\alpha_k\rfloor M^A(W_k).
$$
The matrix $M$ is a CDV matrix for $\P$ with integer coefficients. It remains to prove that the signs of $M_{ij}$ and $M^c_{ij}$ coincide, which would validate that $M$ is positive.
Indeed,
\begin{align*}
|M_{ij} - C M^c_{ij}| 
&= |\sum_{1\le k\le n}(\lfloor C\alpha_k\rfloor - C\alpha_k) M^A_{ij}(W_k)|\\
&\le \sum_{1\le k\le n} |M^A_{ij}(W_k)|
\le 4\max|M^A_{ij}(W_k)|
\le C \min|M^c_{ij}|.
\end{align*}

The third transition holds since exactly 4 wheels participate
in the wheel-decomposition of any single edge.
In the remainder of the proof we bound the entries of $M$.
Due to the bounds on $M^c$ and $M^A$ given by Eq.~\eqref{eq:boundC2} and Eq.~\eqref{eq:boundC1}, the constant $C$ is bounded by
$$
C \le 4\gridsize^{3(\maxdeg-1)+8}.
$$
The constructed CDV matrix $M$ is therefore bounded by
\begin{align*}
|M_{ij}| 
\le& |CM^c_{ij}|+4\max|M^A_{ij}(W_k)| \\
\le& |C|\max|M^c_{ij}|+4\max|M^A_{ij}(W_k)| \\
\le& 4\gridsize^{3(\maxdeg-1)+8} 4\gridsize^2 + 4 \gridsize^{3(\maxdeg-1)}.
\end{align*}
\end{proof}

The final theorem of the paper is now a direct consequence of Theorem~\ref{thm:LoThm} and Theorem~\ref{thm:CDVsimplicial}:
\begin{theorem}
\label{thm:simplicialDuality}
    Let $G$ be a triangulation with maximal vertex degree $\maxdeg$
    and let  $\P=(u_i)_{1\le i\le n}$ be a realization of $G$ as a convex polytope with integer coordinates.
    Then there exists  a realization $(\phi_f)_{f\in F(G)}$ of
    the dual graph $G^*$ as a convex polytope with integer coordinates bounded by
    $$
    |\phi_f|<O(n\max|u_i|^{3\maxdeg + 9}).
    $$
\end{theorem}

\section{Example}

\label{sec:example}
\tikzset{vertex/.style=} 
\tikzset{edge/.style=} 
\tikzset{axis/.style=} 
As an example for our embedding algorithm for truncated polytopes we
show how to embed the \emph{truncated tetrahedron}.
This polytope is obtained from 
the tetrahedron by truncating all of its four vertices.
The graph $G^*$ of this polytope and its dual $G$ are depicted in Fig.~\ref{fig:onlyfig}.
We start with the planar embedding of $G_{\uparrow}=G[v_2,\ldots, v_8]$ that is defined on Fig.~\ref{fig:2d}.

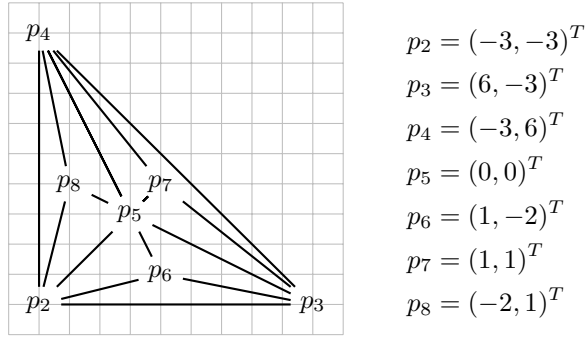
\begin{figure}[ht]
\begin{longtable}{p{4cm}p{4cm}p{4cm}}
\parbox{6cm}{
 \begin{tikzpicture}

  \begin{scope}[scale=0.4]
  \draw[step=1cm,color=lightgray] (-4,-4) grid (7,7);
	\node[vertex] (v2) at (-3,-3) {$p_2$};
	\node[vertex] (v3) at (6,-3) {$p_3$};
	\node[vertex] (v4) at (-3,6) {$p_4$};
	\node[vertex] (v5) at (0,0) {$p_5$};
	\node[vertex] (v6) at (1,-2) {$p_6$};
	\node[vertex] (v7) at (1,1) {$p_7$};
	\node[vertex] (v8) at (-2,1) {$p_8$};

	\draw[edge, thick] (v2) -- (v3) -- (v4) -- (v2);
	\draw[edge, thick] (v3) -- (v7) -- (v4) -- (v8) -- (v2) -- (v6) -- (v3);
	\draw[edge, thick] (v6) -- (v5) -- (v7) -- (v5) -- (v8);
	\draw[edge, thick] (v3) -- (v5) -- (v4) -- (v5) -- (v2);
  \end{scope}
 \end{tikzpicture}
}
&
\parbox{4cm}{
\begin{align*}
p_2&=(-3,-3)^T\\
p_3&=(6,-3)^T\\
p_4&=(-3,6)^T\\
p_5&=(0,0)^T\\
p_6&=(1,-2)^T\\
p_7&=(1,1)^T\\
p_8&=(-2,1)^T
\end{align*}
}
\end{longtable}
\caption{The original plane embedding of  $G_{\uparrow}$.}
\label{fig:2d}
\end{figure}
The embedding has the integer stress 
$\omega_{52}=\omega_{53}=\omega_{54}=1$, $\omega_{23}=\omega_{34}=\omega_{42}=-2$, 
and all other stresses have value $3$.

Following Theorem~\ref{thm:framework}, we embed the drawing of $G_{\uparrow}$ 
onto the plane $\{z=1\}$. 
Unlike the general case of Theorem~\ref{thm:framework}, 
we do not have to scale 3 times and translate, 
since the baricenter of $p_2, p_3, p_4$ is initially at the origin.
We add the point $u_1:=(0,0,-3)^{T}$ 
.

The corresponding CDV matrices $[M''_{ij}]_{1\le i,j\le 4}$ and $[M']_{2\le i,j\le 8}$ are:
\begin{equation}\label{eqn:matrices}
M''=\begin{pmatrix}
1& 1& 1& 1\\
1& 1& 1& 1\\
1& 1& 1& 1\\
1& 1& 1& 1
\end{pmatrix}, \quad
M'=\begin{pmatrix}
-3&	-2&	-2&	1&	3&	0&	3\\
-2&	-3&	-2&	1&	3&	3&	0\\
-2&	-2&	-3&	1&	0&	3&	3\\
1&	1&	1&	-12&	3&	3&	3\\
3&	3&	0&	3&	-9&	0&	0\\
0&	3&	3&	3&	0&	-9&	0\\
3&	0&	3&	3&	0&	0&	-9
\end{pmatrix}.
\end{equation}
We extend $M'$ and $M''$ to the whole $G$ and form the final CDV matrix $M=M'+3M''$:
\begin{equation*}
M=\begin{pmatrix}
3&	3&	3&	3&	0&	0&	0&	0\\
3&	0&	1&	1&	1&	3&	0&	3\\
3&	1&	0&	1&	1&	3&	3&	0\\
3&	1&	1&	0&	1&	0&	3&	3\\
0&	1&	1&	1&	-12&	3&	3&	3\\
0&	3&	3&	0&	3&	-9&	0&	0\\
0&	0&	3&	3&	3&	0&	-9&	0\\
0&	3&	0&	3&	3&	0&	0&	-9
\end{pmatrix}.
\end{equation*}

We can now apply Theorem~\ref{thm:LoThm} and compute the vectors $(\phi_f)$. We first assign $\phi_{(236)}=(0,-18,27)^{T}$.
(One can start with assigning of any vector to any face -- 
the resulting embeddings will be the same up to translation. 
This assignment makes the resulting coordinates look slightly nicer.)
The remaining vectors are then iteratively computed with~\eqref{eq:phi}. We obtain as a result:
{\footnotesize
\begin{align*}
&\quad&&\phi_{(236)}=(0,-18,27)^{T}\\
\phi_{(265)}-\phi_{(236)}&= M_{26}(u_2\times u_6) = (-3,12,27)^{T}	\quad&&\phi_{(265)}=(-3,-6,54)^{T}\\
\phi_{(258)}-\phi_{(265)} &= M_{25}(u_2\times u_5) = (-3,3,0)^{T}	\quad&&\phi_{(258)}=(-6,-3,54)^{T}\\
\phi_{(284)}-\phi_{(258)} &= M_{28}(u_2\times u_8) = (-12,3,-27)^{T}	\quad&&\phi_{(284)}=(-18,0,27)^{T}\\
\phi_{(485)}-\phi_{(284)} &= M_{48}(u_4\times u_8) = (15,3,27)^{T}	\quad&&\phi_{(485)}=(-3,3,54)^{T}\\
\phi_{(457)}-\phi_{(485)} &= M_{45}(u_4\times u_5) = (6,3,0)^{T}	\quad&&\phi_{(457)}=(3,6,54)^{T}\\
\phi_{(473)}-\phi_{(457)} &= M_{47}(u_4\times u_7) = (15,12,-27)^{T}	\quad&&\phi_{(473)}=(18,18,27)^{T}\\
\phi_{(375)}-\phi_{(473)} &= M_{37}(u_3\times u_7) = (-12,-15,27)^{T}	\quad&&\phi_{(375)}=(6,3,54)^{T}\\
\phi_{(356)}-\phi_{(375)} &= M_{35}(u_3\times u_5) = (-3,-6,0)^{T}	\quad&&\phi_{(356)}=(3,-3,54)^{T}\\
\phi_{(362)}-\phi_{(356)} &= M_{36}(u_3\times u_6) = (-3,-15,-27)^{T}	\quad&&\phi_{(362)}=(0,-18,27)^{T}\\
\phi_{(321)}-\phi_{(362)} &= M_{32}(u_3\times u_2) = (0,-9,-27)^{T}	\quad&&\phi_{(321)}=(0,-27,0)^{T}\\
\phi_{(124)}-\phi_{(321)} &= M_{12}(u_1\times u_2) = (-27,27,0)^{T}	\quad&&\phi_{(124)}=(-27,0,0)^{T}\\
\phi_{(143)}-\phi_{(124)} &= M_{14}(u_1\times u_4) = (54,27,0)^{T}	\quad&&\phi_{(143)}=(27,27,0)^{T}.
\end{align*}}
The final result is depicted in Fig.~\ref{fig:final}. The embedding requires a $54\times 54\times 54$ grid.
\begin{figure}[h]
\centering
\begin{tikzpicture}[axis/.style={->,blue}]

	\begin{scope}[xshift=9cm, yshift=-0.7cm, scale=0.06, x={(1cm,0cm)}, y={(0.25cm,0.17cm)}, z={(0cm,1cm)}]
		\draw[axis] (0,0,0) -- (30,0,0) node[anchor=west]{$x$};
	\draw[axis] (0,0,0) -- (0,30,0) node[anchor=west]{$y$};
	\draw[axis] (0,0,0) -- (0,0,30) node[anchor=west]{$z$};

		\path (0,-18,27) node[name=f263, shape=coordinate]{};
		\path (-3,-6,54) node[name=f265, shape=coordinate]{};
		\path  (-6,-3,54) node[name=f258, shape=coordinate]{};
		\path  (-18,0,27) node[name=f284, shape=coordinate]{};
		\path  (-3,3,54) node[name=f485, shape=coordinate]{};
		\path  (3,6,54) node[name=f457, shape=coordinate]{};
		\path  (18,18,27) node[name=f473, shape=coordinate]{};
		\path  (6,3,54) node[name=f375, shape=coordinate]{};
		\path  (3,-3,54) node[name=f356, shape=coordinate]{};
		\path  (0,-18,27) node[name=f362, shape=coordinate]{};
		\path  (0,-27,0) node[name=f321, shape=coordinate]{};
		\path  (-27,0,0) node[name=f124, shape=coordinate]{};
		\path  (27,27,0) node[name=f143, shape=coordinate]{};
		
		\path[draw,thick] (f362) -- (f265) -- (f258) -- (f284) -- (f124) -- (f321);
		\path[draw, dashed] (f284) -- (f485);
		\path[draw,thick] (f485) -- (f457);
		\path[draw, dashed] (f457) -- (f473);
		\path[draw,thick] (f473) -- (f143) -- (f321);
		\path[draw,thick] (f473) -- (f375) -- (f356) -- (f362) -- (f321) -- (f124) ;
		\path[draw, dashed] (f124) -- (f143);
		\path[draw, thick] (f375) -- (f356) -- (f265) -- (f258) -- (f485) -- (f457) -- (f375);
	\end{scope}
\end{tikzpicture}
\caption{The final embedding of the truncated tetrahedron.}
\label{fig:final}
\end{figure}
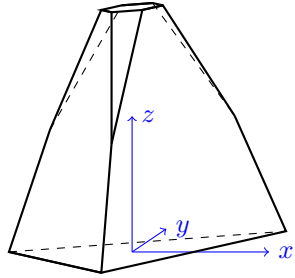

\subsection*{Acknowledgements} We thank an anonymous reviewer for helpful comments that improved the presentation greatly.

\bibliographystyle{abbrv} 
\bibliography{truncated}

\end{document}